\documentclass[12pt]{article}
\usepackage{graphicx}
\usepackage{enumerate}
\usepackage{natbib}
\usepackage{url} 
\usepackage{hyperref}
\usepackage{graphicx, xcolor}
\usepackage{amsmath, amsthm}
\usepackage{amsmath,amssymb}
\usepackage{setspace}
\usepackage[ruled,vlined]{algorithm2e} 
\usepackage{tabularray}
\usepackage{algpseudocode}

\allowdisplaybreaks
\newcommand{\blind}{0}

\newcommand{\possessivecite}[1]{\citeauthor{#1}'s (\citeyear{#1})}

\theoremstyle{definition}
\newtheorem{remark}{Remark}

\newtheorem{corollary}{Corollary}
\newtheorem{assm}{Assumption}
\theoremstyle{plain}
\newtheorem{theorem}{Theorem}

\newtheorem{proposition}{Proposition}

\newcommand{\Var}{\mathbf{Var}}
\newcommand{\E}{\mathbf{E}}
\newcommand{\X}{\mathcal{X}}

\newcommand\numberthis{\addtocounter{equation}{1}\tag{\theequation}}

\begin{document}

\def\spacingset#1{\renewcommand{\baselinestretch}%
{#1}\small\normalsize} \spacingset{1}


\if0\blind
{
  \title{\bf Efficient Multivariate Initial Sequence Estimators for MCMC}
  \author{Arka Banerjee\\
    Department of Mathematics and Statistics,\\ Indian Institute of Technology Kanpur\\ \texttt{arkabee20@iitk.ac.in}\\~\\
    Dootika Vats\thanks{\textit{Dootika Vats is supported by SERB (SPG/2021/001322).}}\hspace{.2cm} \\
    Department of Mathematics and Statistics,\\ Indian Institute of Technology Kanpur\\
    \texttt{dootika@iitk.ac.in}}
  \maketitle
} \fi

\if1\blind
{
  \bigskip
  \bigskip
  \bigskip
  \begin{center}
    {\LARGE\bf Title}
\end{center}
  \medskip
} \fi

\bigskip
\begin{abstract}
Estimating Monte Carlo error is critical to valid simulation results in Markov chain Monte Carlo (MCMC), and initial sequence estimators were one of the first methods introduced for this. Over the last few years, focus has been on multivariate assessment of simulation error, and many multivariate generalizations of univariate methods have been developed. The multivariate initial sequence estimator is known to exhibit superior finite-sample performance compared to its competitors. However, the multivariate initial sequence estimator can be prohibitively slow, limiting its widespread use. We provide an efficient alternative to the multivariate initial sequence estimator that inherits its asymptotic properties as well as the finite-sample superior performance. The effectiveness of the proposed estimator is shown via example implementations. Further, we also present univariate and multivariate initial sequence estimators for when parallel MCMC chains are run and demonstrate their effectiveness over a popular alternative.
\end{abstract}

\noindent%
\vfill

\newpage
\spacingset{1.75} 
\section{Introduction}

Variance assessment for Markov chain Monte Carlo (MCMC) estimators is critical to accurate Bayesian inference. In MCMC, given a target distribution $\pi$ defined on the probability space $(\X, \mathcal{B}(\X))$, an ergodic Markov chain $\{X_t\}_{t\geq 1}$ is constructed such that $\pi$ is its stationary distribution. When interest is in $\theta = \E_{\pi}(g(X))$ for a function $g: \X \to \mathbb{R}^d$, Birkhoff's ergodic theorem allows for the following Monte Carlo estimator to be strongly consistent:
\[
\bar{g}_n := \dfrac{1}{n}  \sum_{t=1}^{n} g(X_t)\,.
\]
Over the years, many output analysis tools have focused on assessing the Monte Carlo error of $\bar{g}_n$. This can be achieved if $g$ and the Markov chain exhibit a central limit theorem (CLT) such that there exists a $d \times d$ positive-definite covariance matrix $\Sigma$ so that as $n \to \infty$,
\begin{equation}
    \sqrt{n}\left(\bar{g}_n - \theta \right) \overset{d}{\rightarrow} \text{N}_d(0, \Sigma). \label{eq:clt}
\end{equation}
Estimating $\Sigma$ is critical for output analysis as it represents the variability in the estimation of $\theta$. Let $\zeta_{i}$ denote the lag-$i$ auto-covariance matrix
\begin{equation*}
    \zeta_{i} = \E_{\pi} \left[\left(g(X_{1}) - \theta \right) \left(g(X_{i+1}) - \theta \right)^{\top} \right]\,.
\end{equation*}
Then, when a CLT holds for $\bar{g}_n$, $\Sigma = \sum_{i=-\infty}^{\infty} \zeta_i$.  It is naturally challenging to estimate $\Sigma$ and many efforts have been made in this direction, both for univariate $(d = 1)$ and multivariate problems. \cite{seila1982multivariate}  proposed regeneration-based estimators that exhibit good properties but are challenging to employ in practice and require studying individual Markov chain kernels. As a consequence, estimators that can be implemented without knowledge of the kernel are more favored. These include batch-means estimators, spectral variance estimators, and initial sequence estimators (ISE). Amongst these, batch-means \citep{chen1987multivariate,jones2006fixed,vats2019multivariate} are computationally efficient but can exhibit significant undesirable underestimation \citep{vats2022lugsail}. Spectral variance estimators \citep{andrews1991heteroskedasticity,flegal2010batch,vats2018strong} can control this underestimation to some extent but are more computationally involved. 

As pointed out in \cite{vats2022lugsail}, in the presence of positive auto-correlation in the Markov chain, underestimation in $\Sigma$ can greatly impact the quality of inference. Specifically, underestimating the asymptotic variance leads to inflated effective sample size estimates, early termination of the MCMC, and a false sense of security about the quality of estimation. 

For reversible Markov chains, univariate ISEs, first proposed by \cite{geyer1992practical}, are known to be asymptotically conservative, exhibiting reduced underestimation in finite samples. ISEs are formed by truncating the sum of the sample auto-covariances, where the truncation point is process-dependent, and is sequentially determined. For the one-dimensional case, all sample auto-covariances can be obtained quickly using a fast Fourier transform (FFT) with a cost of $\mathcal{O}(n\log n)$, alleviating the burden of sequentially calculating sample auto-covariances. \possessivecite{geyer1992practical} ISE has over the years been studied and employed in many applications.  \cite{fabreti2022convergence} demonstrated its superior performance for analyzing the quality of phylogenetic MCMC, and \cite{ardia2017nse,carpenter2017stan,geyer2015package} employ the ISE estimator as default in their MCMC software.  

In most MCMC applications, $g(\cdot)$ is multivariate and marginal assessment of Monte Carlo error is insufficient to capture the complex inter-dependence across the components of $\bar{g}_n$. Similar to batch-means and spectral variance estimators, multivariate extensions of \possessivecite{geyer1992practical} ISE have been proposed by \cite{dai2017multivariate} and \cite{kosorok2000monte}. However, element-wise univariate FFTs add significant memory burden without considerable time gains and require repeated determinant calculations to obtain the truncation point. The resulting estimator is prohibitively slow.

We propose a new multivariate version of the ISE estimator which combines the finite-sample performance of the ISE with the computational efficiency of batch-means estimators while retaining the asymptotic theoretical properties of the multivariate estimator of \cite{dai2017multivariate}.  Section~\ref{sec:new_ISE} elaborates on the proposed estimator, its computational cost, and its theoretical properties. Running parallel Markov chains is commonplace in MCMC. However, univariate or multivariate ISEs have not been developed for this framework. In Section~\ref{sec:parallel_ise} we present a parallel-chain version of \possessivecite{geyer1992practical} ISE and then adapt it for a computationally efficient multivariate estimator. We show that the parallel-chain version of the ISE estimator retains the asymptotic properties of the single-chain estimator. Our examples in  Section~\ref{sec:examples} highlight the improved estimation quality and the drastically reduced computation time for our method compared to that of \cite{dai2017multivariate}. We first compare the performance of our proposed estimator to the existing ones for a toy problem where the true value of $\Sigma$ is known. Our second example is that of an MCMC employed in a Bayesian Poisson regression model analyzing spike train data\footnote{Code is available at \href{https://github.com/Arkagit/Efficient-Initial-Sequence-estimator}{https://github.com/Arkagit/Efficient-Initial-Sequence-estimator}}. In both our implementations it is evident that our proposed estimator, in addition to being significantly faster than current multivariate ISE, exhibits improved finite-sample properties.  

\section{Initial Sequence Estimators}
\label{sec:ISE}
    \subsection{Univariate Initial Sequence Estimators}
    \label{sec:defn and back}
        We present the original positive ISE of \cite{geyer1992practical} for the univariate case. For $d = 1$, if $\gamma_i$ denotes the lag-$i$ auto-covariance,
        \begin{equation*}
            \gamma_{i} = \text{Cov}_{\pi}(g(X_1), g(X_{i+1}))\,,
        \end{equation*}
        then the asymptotic variance in the Markov chain CLT in \eqref{eq:clt} can be expressed as
        \begin{equation*}
            \sigma^2 = -\gamma_{0} + 2 \sum_{i=0}^{\infty} \left( \gamma_{2i} + \gamma_{2i + 1} \right)\,.
        \end{equation*}
        For $k = 0, \dots, n-1$, denote the sample auto-covariance as
        \begin{equation*}
            \hat{\gamma}_{i;n} := \dfrac{1}{n} \sum_{t=1}^{n-i} (g(X_t) - \bar{g}_n) (g(X_{t+i}) - \bar{g}_n)\,.
        \end{equation*}
        Ideally, we would like to estimate $\sigma^2$ using the sum of all $\hat{\gamma}_{i;n}$'s, however, this leads to an inconsistent estimator. Define $\Gamma_{i} := (\gamma_{2i} + \gamma_{2i+1})$ for all $i \ge 0$ and let $\hat{\Gamma}_{i;n} = (\hat{\gamma}_{2i;n} + \hat{\gamma}_{2i+1;n})$. For reversible Markov chains, \cite{geyer1992practical} showed that $\Gamma_{i} > 0$ for all $i$. This result is then used to truncate the sum of the $\hat{\gamma}_{i;n}$, yielding the ISE. That is,
        \begin{align*}
            \hat{\sigma}^{2}_{\text{ISE}} 
            = - \hat{\gamma}_{0;n} + 2 \sum_{i=0}^{k_{n}} \hat{\Gamma}_{i;n}\,,
        \end{align*}
        where $k_{n}$ is the largest integer such that $\hat{\Gamma}_{i;n} > 0$ for all $i \in \{1, 2, \ldots, k_{n}\}$. That is, the sum of sample auto-covariances is stopped as soon as there is evidence that the auto-covariance for that lag is not well estimated. 
        \begin{theorem}{\citep[][Theorem~3.2]{geyer1992practical}}
        \label{thm:geyer_ise}
            For a $\pi$-reversible and Harris ergodic stationary Markov chain such that a univariate CLT holds for $g$, 
            \begin{equation*}
                \liminf_{n \rightarrow \infty} \hat{\sigma}^{2}_{\normalfont{ISE}} \ge \sigma^{2} \qquad{\text{with probability $1$}}\,.
            \end{equation*}
        \end{theorem}

        Theorem~\ref{thm:geyer_ise} establishes asymptotic conservativeness of the $\hat{\sigma}^{2}_{\text{ISE}}$, and numerical experiments through the years have also established good finite-sample properties. Particularly for slowly mixing Markov chains for which truncation time $k_n$ can be large,  implementation of the ISE is greatly benefited by an FFT which allows for a calculation of all $\hat{\gamma}_{i:n}, i = 0, 1, 2, \dots, n-1$ in $\mathcal{O}(n \log n)$ time. On the flip side, the strong consistency of the ISE has not been shown yet and remains an open question. \cite{berg2023efficient} recently proposed a shape-constrained estimator of the sequence of $\gamma_{i}$, which they then used to obtain a moment least squares estimator (MLS) of $\sigma^2$ and obtained conditions for strong consistency. The computational complexity of the estimator is unclear, but in our simulations, it is evident that the MLS estimator requires more time than $\hat{\sigma}^{2}_{\text{ISE}}$.

    \subsection{Multivariate Initial Sequence Estimators}

        MCMC samplers are usually employed when $\pi$ is multidimensional and exhibits complex dependence across components. Consequently, multivariate error assessment in $\bar{g}_n$ is critical to MCMC output analysis. \cite{kosorok2000monte} first extended \possessivecite{geyer1992practical} univariate ISE to a multivariate case, and an improved alternative was presented by \cite{dai2017multivariate}. Although built in the same spirit as the univariate ISE, the practical limitations of the multivariate initial sequence estimators (mISE) are significant.  Let $\hat{\zeta}_{i;n}$ denote the lag-$i$ sample auto-covariance matrix for $i = 0, 1, \ldots, (n-1)$
        \begin{equation*}
            \hat{\zeta}_{i;n} = \hat{\zeta}_{-i;n}^{\top} = \frac{1}{n} \sum_{t = 1}^{n-i} (g(X_{t}) - \bar{g}_n)(g(X_{t+i}) - \bar{g}_n)^{\top}.
        \end{equation*}
        Let $\hat{Z}_{i;n} = (\hat{\zeta}_{2i;n} + \hat{\zeta}_{- 2i;n}^{\top})/{2} + (\hat{\zeta}_{2i+1;n} + \hat{\zeta}_{- (2i + 1);n})/{2}$ and for $0 \leq m \leq \lfloor n/2 -1\rfloor$, define
        \begin{equation*}
            \hat{\Sigma}_{m;n} = - \hat{\zeta}_{0;n} + 2 \sum_{i=0}^{m} \hat{Z}_{i;n}.
        \end{equation*}
        Here, $\hat{\Sigma}_{m;n}$ is the sum of sample auto-covariances up to lag $2m$. \cite{dai2017multivariate} showed that the population quantity $Z_i = \zeta_{2i} + \zeta_{2i+1}$ is positive-definite for all $i$ and further $Z_{i} - Z_{i-1}$ is also positive-definite for all $i \geq 1$. Using these properties, \cite{dai2017multivariate} define their mISE as
        \begin{equation*}
            \hat{\Sigma}_{\text{ISE}} := \hat{\Sigma}_{t_{n};n} \,,
        \end{equation*}
        where $t_{n}$ is the truncation time as described in Algorithm~\ref{algo:mise}.
  
    \begin{algorithm}
        \caption{mISE of \cite{dai2017multivariate}} \label{algo:mise}
        \KwData{$\{g(X_{1}), \ldots, g(X_{n})\}$}
        $\hat{\Sigma} \gets - \hat{\zeta}_{0;n}$; \Comment{Cost: $\mathcal{O}(d^{2}n)$}\\
        \For{$i = 0, 1, \ldots, \lfloor (n-1)/2 \rfloor$}{
            $\hat{\Sigma} = \hat{\Sigma} + 2\hat{Z}_{i;n}$;\\
            \If {\normalfont {$\hat{\Sigma}$ is positive-definite}} {
                $s_{n} = i$;
                \textbf{break};
            } \Comment{Cost: $\mathcal{O}(d^{2}n s_{n})$}
        }
        \For{$i = (s_{n} + 1), \ldots, \lfloor n/2-1 \rfloor$}{
            $S = \hat{\Sigma}$;\\
            Calculate $\hat{Z}_{i;n}$; \Comment{Cost: $\mathcal{O}(d^2n t_n)$}\\
            
            $\hat{\Sigma} = \hat{\Sigma} + 2\hat{Z}_{i;n}$;\\
            Calculate \text{det}($\hat{\Sigma}$) and \text{det}($\hat{S}$)  \Comment{Cost: $\mathcal{O}(d^3 t_n)$}\\
            \If {\normalfont {\text{det}($\hat{\Sigma}$) $\leq$ \text{det}($S$)} }{
                $t_{n} = (i-1)$;   \Comment{The truncation time} \\ 
                $\hat{\Sigma}_{\text{ISE}} \gets S$;
                \textbf{break};
            } 
        }
        \KwRet {\normalfont {$\hat{\Sigma}_{\text{ISE}}$}} \Comment{Total Cost: $\mathcal{O}(d^{2} n t_{n} + d^3t_n)$}\\ 
    \end{algorithm}

    Both the univariate ISE and mISE do not require any tuning. However, for slow mixing Markov chains, $t_n$ can be quite large, and further, as $n$ increases, the quality of estimation of $Z_i$ improves, increasing the truncation time $t_n$. As illustrated in Algorithm~\ref{algo:mise}, the overall cost of the $\hat{\Sigma}_{\text{ISE}}$ estimator is $\mathcal{O}(d^{2} n t_{n} + d^3t_n)$. The following proposition establishes that $t_n \to \infty$ as $n \to \infty$, yielding both $\mathcal{O}(d^2 n t_n)$ and $\mathcal{O}(d^3 t_n)$ terms to be significant.
    \begin{proposition}
        \label{thm:mise_tn}
        Let $t_{n}$ be the truncation time for the mISE in Algorithm~\ref{algo:mise}. Then with probability $1$
        \begin{equation}
            t_{n} \rightarrow \infty \quad \text{ as } n \rightarrow \infty\,.
        \end{equation}
    \end{proposition}

    \begin{proof}
        See Appendix~\ref{pf:mise_tn}.
    \end{proof}

    Proposition~\ref{thm:mise_tn} implies that for long-run MCMC samplers, the truncation time can be large. Although the rate of increase of $t_n$ is not known, in Section~\ref{sec:examples}, we empirically show that $t_n$ is polynomial in $n$, and thus dominates $\log(n)$. Further, Algorithm~\ref{algo:mise} requires repeated and sequential calculation of determinants and thus is not parallelizable. Consequently, the mISE of \cite{dai2017multivariate} can be prohibitively slow, impeding its widespread use.

    \begin{remark}\label{rem:mise_remark}
    In Appendix~\ref{app:fft} we present an alternative implementation of the mISE, one that is obtained by calling a univariate FFT for each element of the matrix. Although, the computational complexity of this implementation is $\mathcal{O}(d^2 n \log n + d^3 t_n)$, saving all the auto-covariance matrices adds a $\mathcal{O}(d^2n)$ burden on the memory, one that can be quite prohibitive, particularly in high-dimensions. Thus, we recommend the non-FFT implementation in Algorithm~\ref{algo:mise}.
    \end{remark}
 


\section{Covariance-Correlation ISE}
\label{sec:new_ISE}
The central idea for the proposed estimator is based on the observation that any covariance matrix can be decomposed into the product of a diagonal matrix of marginal standard deviations and its corresponding correlation matrix. That is, $\Sigma$ can be written as
\begin{equation}
    \Sigma = L \mathcal{R} L^{\top}\,, \label{eq:obs}
\end{equation}
where $\mathcal{R}$ is the correlation matrix corresponding to $\Sigma$ and $L$ is a diagonal matrix of marginal standard deviations. Let $r_{ij}$ denote the $(i, j)^{\text{th}}$ element of $\mathcal{R}$, and let $\sigma_{ij}$ denote the $(i, j)^{\text{th}}$ element of $\Sigma$. For the diagonals of $\Sigma$, we also adopt the notation $(\sigma^{(i)})^2 = \sigma_{ii}$. We will adopt different methods for estimating $L$ and $\mathcal{R}$. In order to do this, we first present the multivariate batch-means estimator of $\Sigma$. 

For a Monte Carlo sample size of $n = a_n b_n$, let $b_{n}$ be the batch-size and $a_{n}$ be the number of batches. For $k = 1, 2, \dots, a_n$, denote the mean vector in each batch by $\tilde{g}_{k} := b_n^{-1} \sum_{t = (k-1)b_{n} + 1}^{k b_{n}} g(X_{t})$.
    %
    %
    The batch-means estimator of $\Sigma$ first presented by \cite{chen1987multivariate} is defined as
    \begin{equation}
    \label{eq:bm_est}
        \hat{\Sigma}_{\text{BM}} := \frac{b_{n}}{a_{n} - 1} \sum_{k=1}^{a_{n}} (\tilde{g}_{k} - \Bar{g}_n)(\tilde{g}_{k} - \Bar{g}_n)^{\top}.
    \end{equation}
    Let $\left(\hat{\sigma}_{\text{BM}}^{(s)}\right)^{2}$ denote the $s^{\text{th}}$ diagonal element of $\hat{\Sigma}_{\text{BM}}$ and set
    \begin{equation}
    \label{eq:bm_diag}
        \hat{L}_{\text{BM}} = \text{diagonal} \left(\hat{\sigma}_{\text{BM}}^{(1)}, \hat{\sigma}_{\text{BM}}^{(2)}, \dots, \hat{\sigma}_{\text{BM}}^{(d)}  \right)\,.
    \end{equation}
    Using \eqref{eq:bm_est} and \eqref{eq:bm_diag}, we can obtain a batch-means estimator of the correlation matrix,
    \begin{equation}
        \hat{\mathcal{R}}_{\text{BM}} = \hat{L}_{\text{BM}}^{-1} \hat{\Sigma}_{\text{BM}} \hat{L}_{\text{BM}}^{-1}\,.
    \end{equation}
In the presence of positive autocorrelation in the components of the process $\{g(X_t)\}$, for any $s^{\text{th}}$ marginal, $\hat{\sigma}_{\text{BM}}^{(s)}$ underestimates $\sigma^{(s)}$ \citep{flegal2010batch}, with the underestimation being worse when autocorrelation is high. Underestimating these marginal variances has a significant impact on the quality of inference. 
However, empirically, we notice that the relative bias in the estimation of the correlation matrix can be considerably smaller than the relative bias in the estimation of the covariance matrix; we present simulations in support of this in Appendix~\ref{sec:underestimation}. 
Given this observation, we propose to utilize $\hat{\mathcal{R}}_{\text{BM}}$ to estimate $\mathcal{R}$, and estimate $L$ using an ISE estimator.

Let $\hat{L}_{\text{ISE}}$ denote the diagonal matrix of \possessivecite{geyer1992practical} ISE standard deviations
    \begin{equation}
    \label{eq:ise_diag}
        \hat{L}_{\text{ISE}} = \text{diagonal} \left(\hat{\sigma}_{\text{ISE}}^{(1)}, \hat{\sigma}_{\text{ISE}}^{(2)}, \dots, \hat{\sigma}_{\text{ISE}}^{(d)}  \right)\,.
    \end{equation}
Using $\hat{\mathcal{R}}_{\text{BM}}$ and $\hat{L}_{\text{ISE}}$ we obtain our proposed covariance-correlation ISE estimator of $\Sigma$
\begin{equation}
    \hat{\Sigma}_{\text{cc}} = \hat{L}_{\text{ISE}} \hat{\mathcal{R}}_{\text{BM}} \hat{L}_{\text{ISE}}\,.
\end{equation}
%
    %
%
The construction and computational workflow of $\hat{\Sigma}_{\text{cc}}$ is explained in Algorithm~\ref{algo:cc_estimator}. Since the marginal univariate ISEs are efficient, obtaining the matrix $\hat{L}_{\text{ISE}}$ is $\mathcal{O}(d n \log n)$, and since batch-means estimators are fairly efficient, $\hat{\mathcal{R}}_{\text{BM}}$ also has low computational burden. This yields an overall computational cost for $\hat{\Sigma}_{\text{cc}}$ to be $\mathcal{O}(d^2a_n+ dn \log n)$, and in Section~\ref{sec:examples} we show that the observed computational cost is significantly lower than mISE. We follow the recommendation of \cite{vats2021batch} and employ batch size $b_n \propto n^{1/3}$, so that the computational complexity of $\hat{\Sigma}_{\text{cc}}$ for our implementation is $\mathcal{O}(d^2 n^{2/3} + dn \log n)$. Compared to the complexity of the $\hat{\Sigma}_{\text{ISE}}$ estimator, $\hat{\Sigma}_{\text{cc}}$ scales significantly better in dimension. In our simulations, we demonstrate that $t_n$ scales slower than $\log n$, which then implies that $\hat{\Sigma}_{\text{cc}}$ also scales better in $n$ as well.

\begin{algorithm}
        \caption{Proposed CC-ISE estimator with batch size $b_n$} 
        \label{algo:cc_estimator}
        \KwData{$\{g(X_{1}), \ldots, g(X_{n})\}$}
        Calculate $\hat{\Sigma}_{\text{BM}}$ \Comment{Cost: $\mathcal{O}(d^2a_n + dn)$}\\
        Calculate $\hat{\mathcal{R}}_{\text{BM}}$ from $\hat{\Sigma}_{\text{BM}}$ \Comment{Cost: $\mathcal{O}(d^2)$}\\ 
        Calculate $\hat{L}_{\text{ISE}}$ \Comment{Cost: $\mathcal{O}(d n \log n)$}\\ 
        Set $\hat{\Sigma}_{\text{cc}} = \hat{L}_{\text{ISE}} \hat{\mathcal{R}}_{\text{BM}}  \hat{L}_{\text{ISE}}$ \Comment{Cost: $\mathcal{O}(d^2)$}\\ 
        \KwRet {\normalfont {$\hat{\Sigma}_{\text{cc}}$}} \Comment{Total Cost: $\mathcal{O}(d^{2}a_n + dn \log n)$}\\ 
    \end{algorithm}

The marginal variances are well estimated by the ISE, and the resultant estimator is guaranteed to be positive semi-definite. 
Later in Section~\ref{sec:examples} we demonstrate these finite-sample properties of $\hat{\Sigma}_{\text{cc}}$ and here provide conditions under which $\hat{\Sigma}_{\text{cc}}$ is asymptotically conservative. This naturally requires conditions ensuring asymptotic properties of batch-means estimators. 
\begin{assm}
    \label{assm:1} Let $\Vert \cdot \Vert$ denote the Euclidean norm and let $\{B(t): t \ge 0\}$ denote a $d$-dimensional standard Brownian motion. We assume a strong invariance principle holds in that there exists  $d \times d$  positive-definite matrix $M$ such that $\Sigma = MM^{\top}$, a univariate positive random variable $D$ and a positive increasing function $\kappa(n)$  such that as $n \rightarrow \infty$, with probability $1$
\begin{equation}
    \left \Vert \sum_{t=1}^{n} g(X_t) - n \theta - M B(n) \right\Vert < D \kappa(n)\,.
\end{equation}
\end{assm}

There has been much work done in ensuring a strong invariance principle holds for MCMC \citep{jones2006fixed,li2024multivariate}. \cite{vats2018strong} show that $\kappa(n) = n^{1/2 - \lambda}$ for some $0 < \lambda < 1/2$ for when $\{X_t\}$ is polynomially ergodic. \cite{banerjee2023multivariate} show $\kappa(n) = n^{\beta} \log n$ for $\beta = \max\{1/(2 + \delta), 1/4\}$ for when the process is either geometrically or polynomially ergodic with moment conditions on $g$. 
\begin{assm}
    \label{assm:2} The batch-size $b_{n}$ should be chosen so that
\begin{itemize} 
    \item[1.] $b_{n}$ and $n/b_{n}$ are increasing integer sequences such that as $n \rightarrow \infty$, $b_{n} \rightarrow \infty$ and $n/b_{n} \rightarrow \infty$.
    \item[2.] there exists $c \ge 1$ such that $\sum_{n=1}^{\infty} (b_{n}/n)^{c} < \infty$.
    \item[3.] $b_{n}^{-1} \log(n) \kappa^{2}(n) \rightarrow 0$ as $n \rightarrow \infty$.
\end{itemize}
\end{assm}
For the batch-means estimator to be strongly consistent, the number of batches and batch-size both must increase to infinity. Further, the batch-size is typically governed by the strength of the correlation 
 in the Markov chain. It is common to use batch-sizes of the form $b_n = \lfloor n^{\nu} \rfloor$ for some $0 < \nu < 1$, and we employ the recommended batch size of $b_n \propto n^{1/3}$ of \cite{vats2021batch}.


\begin{theorem}
    \label{thm:Gen_var}
    Let $\{X_{t}\}_{t\ge 1}$ be a $\pi$-reversible Harris ergodic stationary Markov chain such that Assumption~\ref{assm:1} holds for $g:\mathcal{X} \to \mathbb{R}^d$. If Assumption~\ref{assm:2} holds, then we have
    \begin{enumerate}[(a)]
        \item $\liminf_{n \rightarrow \infty} \text{det}\left( \hat{\Sigma}_{\text{cc}} \right) \ge \text{det}\left( \Sigma \right) \text{   with probability  } 1, \text{  and  }$
        \item for $\hat{a}_{ij}$ and $a_{ij}$ denoting the $(i, j)^{\text{th}}$ element of $\hat{\Sigma}_{cc}$ and $\Sigma$, respectively,
        \begin{equation*}
            \liminf_{n \rightarrow \infty} \vert \hat{a}_{ij} \vert \ge \vert a_{ij} \vert \text{   with probability  } 1.
        \end{equation*}
    \end{enumerate}
    
\end{theorem}

\begin{proof}
    See Appendix~\ref{pf:Gen_var}.
\end{proof}



\begin{remark}
    Theorem~\ref{thm:Gen_var} establishes asymptotic conservativeness of the generalized variance of $ \hat{\Sigma}_{\text{cc}}$, a result similar to that of mISE from \cite{dai2017multivariate}. The result implies that asymptotically,  the confidence regions created from $\hat{\Sigma}_{\text{cc}}$ will not have a smaller volume than the oracle. Although generalized variance has its statistical importance, it fails to comment on the estimation of all the cross-covariance terms of $\Sigma$. Part $(b)$ of the theorem establishes element-wise asymptotic conservativeness.
    \end{remark}

    \begin{remark}
    Similar to \cite{dai2017multivariate,geyer1992practical,kosorok2000monte, vats2019multivariate}, our asymptotic results assume stationarity of the Markov chain. Starting from any arbitrary distribution introduces a period of initial transience.   \cite{alex2007replicatedbm} and \cite{ockerman1997impact} discuss the effect of this initial transience on variance estimators and present some empirical insights and possible corrections. Establishing asymptotic results for any arbitrary distribution remains an open problem for univariate and multivariate initial sequence estimators, as well as batch-means estimators.
\end{remark}


The univariate MLS estimator of \cite{berg2023efficient} can be modified similarly for the multivariate setting. We define $\hat{\Sigma}_{\text{MLS}} := \hat{L}_{\text{MLS}} \hat{\mathcal{R}}_{\text{BM}} \hat{L}_{\text{MLS}}^{\top}$ where $\hat{L}_{\text{MLS}} = \text{diagonal}(\hat{\sigma}_{\text{MLS}}^{(i)}; 1 \leq i \leq d)$. Under the conditions ensuring strong consistency of the MLS estimator from \cite{berg2023efficient}, $\hat{\Sigma}_{\text{MLS}}$ can also be shown to be strongly consistent. However, we will show in Section~\ref{sec:examples} that this estimator, although faster than the mISE, is significantly slower than $\hat{\Sigma}_{\text{cc}}$.

\section{Numerical Implementations}
\label{sec:examples}

\subsection{Vector Autoregressive Process}

We study the comparative performance of the estimators in a benchmark vector autoregressive (VAR) process example. The VAR example yields a setting where the true $\Sigma$ is known. Consider a VAR process of order $1$, $\{X_{t}\}_{t\ge 1}$,  such that for a $d \times d$ matrix $\Phi$ and $t = 1, 2, \ldots$
\begin{equation*}
    X_{t} = \Phi X_{t-1} + \epsilon_{t}\,,
\end{equation*}
where $\epsilon_{t} \overset{\text{iid}}{\sim} \text{Normal}(0, \Omega)$, $\Omega$ being a $d \times d$ positive-definite matrix. Let matrix $V$ be such that $\textsc{vec}(V) = (I_{d^{2}} - \Phi \otimes\Phi) \textsc{vec}(\Omega)$ where $\otimes$ denotes Kronecker product. Then, if the spectral norm of $\Phi$ is less than 1, $\text{Normal}(0, V)$ is the invariant distribution for $\{X_{t}\}_{t\ge 1}$ and \cite{osawa1988reversibility} showed that the process is reversible when $\Phi \Omega$ is symmetric. For estimation of $\E_{\pi}(X) = 0$, we consider $\bar{g}_n = n^{-1} \sum_{t=1}^{n} X_t$. The asymptotic covariance is given by
\begin{equation*}
    \Sigma_{\text{True}} = (I_{d} - \Phi)^{-1} V + V (I_{d} - \Phi)^{-1} - V.
\end{equation*}

We borrow the setup of \cite{dai2017multivariate} and choose $\Omega = I_{d}$ and $\Phi = d^{-1} \mathcal{H}_{d} $ $\text{diag}(\rho^{-1}, \rho^{-2}, \ldots, \rho^{-d})\mathcal{H}_{d}^{\top}$, where $\mathcal{H}_{d}$ is the Hadamard matrix of dimension $d = 12$. With this choice of $\Phi$ and $\Omega$, the process $\{X_{t}\}_{t\ge 1}$ is reversible. For $\rho = 1.01$, we compare the performance of $\hat{\Sigma}_{\text{cc}}$ with existing estimators. For $n = 5\times 10^{3} \text{  to  } 5 \times 10^{5}$, we assess the performance of estimators of $\Sigma$  by four metrics (i) estimated effective sample size (ESS), (ii) relative Frobenius norm, (iii) computation time, and (iv) coverage probability, all averaged over 1000 replications. \cite{vats2019multivariate} defines the estimated ESS as
\begin{equation*}
    \widehat{\text{ESS}} = n \left(\frac{\text{det}(\hat{\zeta}_{0})}{\text{det}(\hat{\Sigma})}\right)^{1/d},
\end{equation*}
where $\hat{\zeta}_{0}$ is the sample covariance matrix and $\hat{\Sigma}$ is the estimator of interest for $\Sigma$. \cite{vats2021revisiting} showed that there is a one-to-one correspondence between ESS and the $R$-hat of \cite{gelman1992inference}. Consequently, stopping simulation when $\widehat{\text{ESS}}$ is larger than a pre-specified lower bound is standard, and thus it is critical to ensure estimators of $\Sigma$ do not underestimate $\text{det}(\Sigma)$.
For $\| \cdot \|_{\text{F}}$ denoting Frobenius norm, the relative Frobenius norm is
\begin{equation*}
    \text{Relative Frobenius Norm} = \frac{\Vert \Sigma - \hat{\Sigma}\Vert_{\text{F}}}{\Vert \Sigma \Vert_{\text{F}}}.
\end{equation*}

\begin{figure}
  \begin{minipage}[b]{0.5\linewidth}
    \centering
    \includegraphics[width=.9\linewidth]{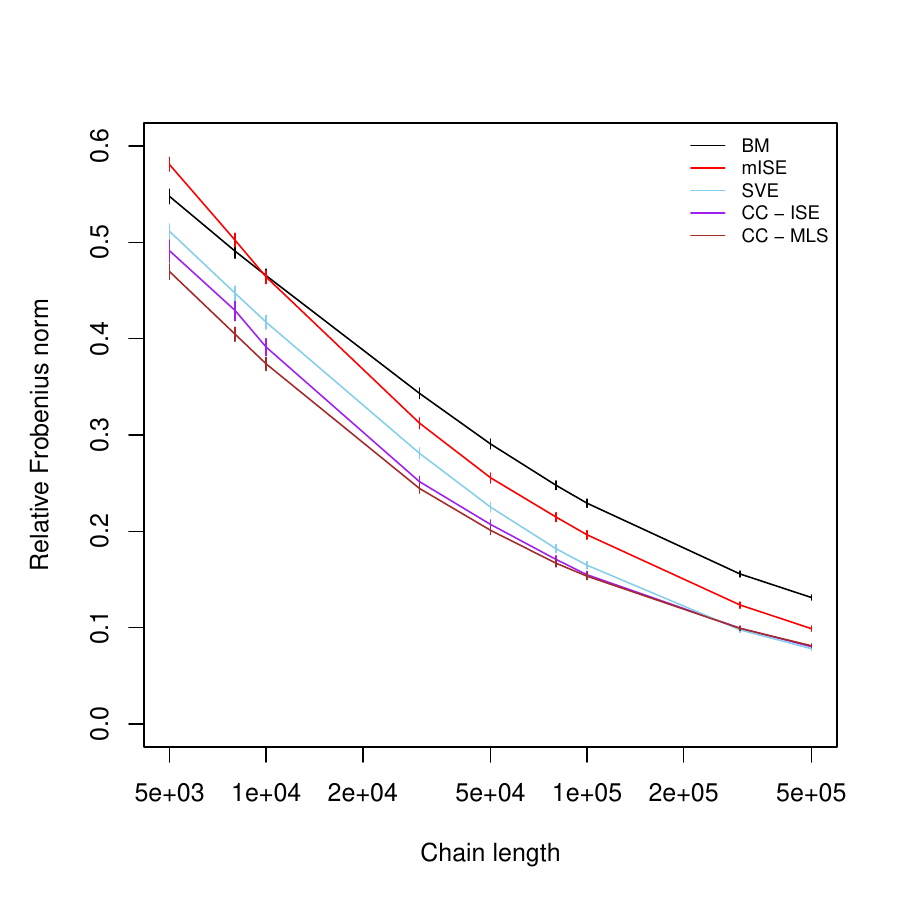} 
    \vspace{4ex}
  \end{minipage} 
  \begin{minipage}[b]{0.5\linewidth}
    \centering
    \includegraphics[width=.9\linewidth]{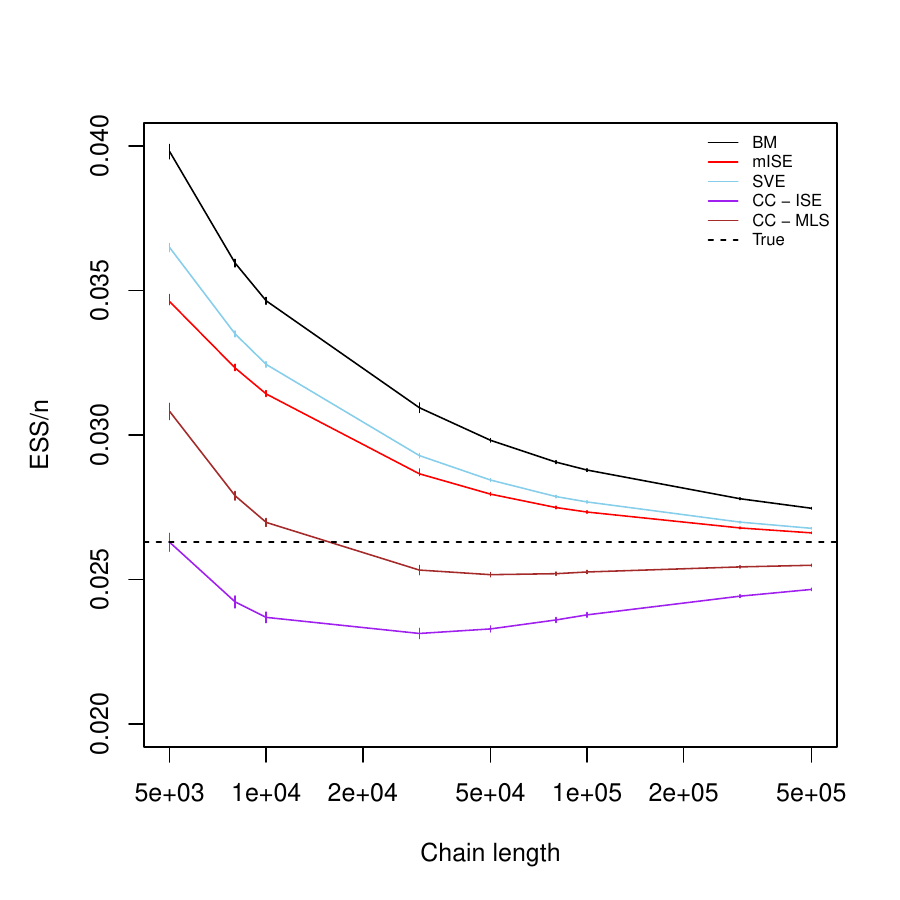} 
    \vspace{4ex}
  \end{minipage}\\
  \begin{minipage}[b]{0.5\linewidth}
    \centering
    \includegraphics[width=.9\linewidth]{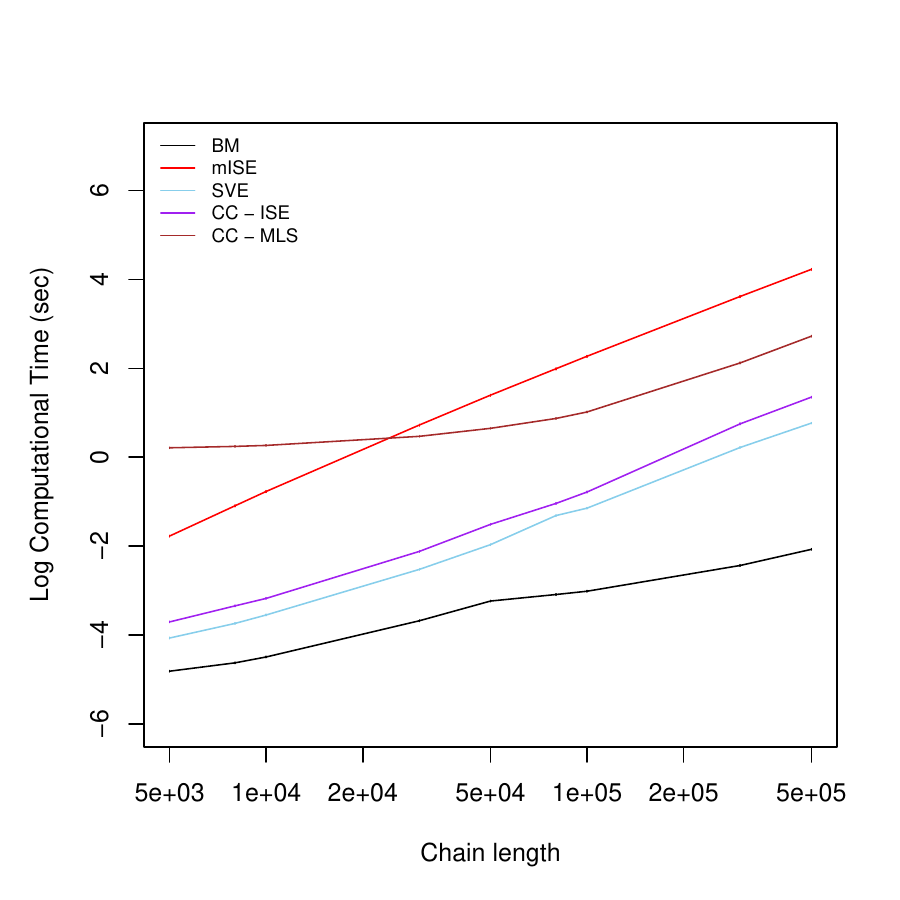} 
    \vspace{4ex}
  \end{minipage} 
  \begin{minipage}[b]{0.5\linewidth}
    \centering
    \includegraphics[width=.9\linewidth]{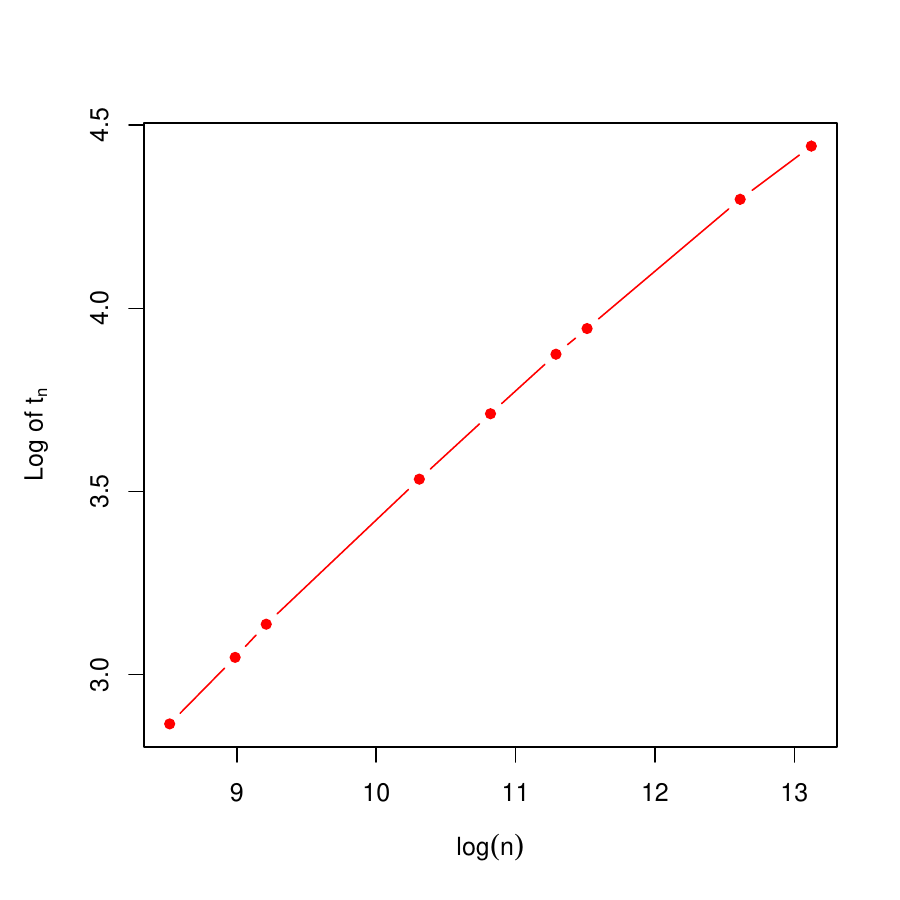} 
    \vspace{4ex}
  \end{minipage} 

  \caption{(Top left) Relative Frobenius norm of estimators of $\Sigma$. (Top right) Estimated ESS$/n$. (Bottom left) Log computation time in seconds. (Bottom right) Log of truncation time ($\log t_n$) for mISE as compared to $\log n$, indicating a polynomial rate for $t_n$.}
  
  \label{fig:var_plots} 
\end{figure}

We compare the $\hat{\Sigma}_{\text{cc}}$ (CC-ISE) with the \possessivecite{dai2017multivariate} mISE estimator, along with batch-means (BM) and spectral variance estimator (SVE). Additionally, we also implement CC-MLS, the covariance-correlation estimator built with \possessivecite{berg2023efficient} MLS estimator for the marginal variances. The results are provided in Figure~\ref{fig:var_plots}.  First, both the CC estimators yield the lowest relative Frobenius norm, significantly lower than the mISE estimator. For ESS, as explained earlier, in order to avoid early termination of simulation, underestimating the true quantity is more favorable than overestimating; here we see clearly that both the CC estimators converge to the true value from below. These gains in finite-sample performance are supported by quick computation. The mISE estimator is prohibitively slow, particularly for the large sample sizes typically found in MCMC. In contrast, $\hat{\Sigma}_{\text{cc}}$ is much faster, only a little slower than SVE. The CC-MLS estimator is slower than CC-ISE since the marginal variance calculations are significantly slower.

Over the 1000 replications, we also track the truncation time for the mISE estimator for varying sample sizes. In the bottom right plot of Figure~\ref{fig:var_plots}, we demonstrate the rate of increase of $\log t_n$ for the mISE estimator as it compares to $\log n$. The plot shows that $\log t_{n}$ is linear in $\log n$ and thus $t_n$ is likely to be polynomial in $n$. Thus, in this example, $t_n$ certainly dominates $\log n$ and diverges to infinity as chain length increases.

For assessing the coverage probability, we determine whether the asymptotic confidence ellipsoid contains the true value of $\E_{\pi}X = 0$. An asymptotic $95\%$ confidence ellipsoid from \eqref{eq: corr_cov}, using an estimator $\hat{\Sigma}$ of $\Sigma$, is
\begin{equation*}
    {C}_{0.95} = \left\{\mu \in \mathbb{R}^{d} : n(\bar{g}_n - \mu)^{\top} \hat{\Sigma}^{-1}(\bar{g}_n - \mu) < \chi^{2}_{d, 0.95} \right\}.
\end{equation*}
Over 1000 replications and for competing estimators of $\Sigma$, we determine the number of times the estimated confidence ellipsoid contains 0. Results are in Table~\ref{table:VAR_coverage}. As expected, BM and SVE have low coverage for smaller sample sizes due to the underestimation of the determinant of $\Sigma$. \possessivecite{dai2017multivariate} mISE also has low coverage at smaller sample sizes, indicating an underestimation of the target variance. On the other hand, CC estimators exhibit improved coverage at low sample sizes as well as large sample sizes, with CC-ISE showing better performance than CC-MLS.

\begin{table}
\centering
\begin{tabular}{ |p{2cm}|ccccc|  }
 
 \hline
   Method & \multicolumn{5}{|c|}{Chain Length} \\
 \hline
& $5 \times 10^{3}$ & $10^{4}$ & $5 \times 10^{4}$ & $10^5$ & $5\times10^{5}$\\
 \hline
 BM & $0.474$ & $0.664$ & $0.883$ & $0.887$ & $0.952$\\
 SVE & $0.589$ & $0.751$ & $0.896$ & $0.913$ & $0.960$ \\
 mISE& $0.651$ & $0.778$ & $0.906$ & $0.913$ & $0.960$\\
CC-ISE& $0.715$ & $0.883$ & $0.948$ & $0.962$ & $0.974$\\
CC-MLS & $0.623$ & $0.824$ & $0.936$ & $0.944$ & $0.971$\\
 \hline

\end{tabular}
 \caption{Coverage probability at $95 \%$ confidence for different chain lengths.}
 \label{table:VAR_coverage}
\end{table}

%

\subsection{Bayesian Poisson Regression for Spike Train Data}

Consider the spike train dataset from \cite{dangelo2023canale} which uses calcium imaging to observe neuron responses to external events. The data consists of $p = 23$ factors as predictors for 920 responses of the count of activations in each neuron. \cite{dangelo2023canale} use the following Bayesian Poisson regression model to analyze this data.

For $i = 1, 2, \dots, 920$, let $Y_i$ denote the random response and $y_i$ be the realized value of the count of the neuron activations. Further, let  $\mathbf{x}_{i} \in \mathbb{R}^{23}$ denote the $i^{\text{th}}$ vector of predictors. For $\beta \in \mathbb{R}^{23}$, \cite{dangelo2023canale} propose the following model,
\begin{align*}
    Y_{i}\mid \beta & \overset{\text{ind}}{\sim} \text{Poisson} \left( \exp(\mathbf{x}^{\top}_{i} \beta) \right) \\ 
    \beta & \sim \text{Normal} (\mathbf{b}, B).
\end{align*}
The posterior distribution for $\beta$ is unavailable and MCMC is employed to obtain posterior estimates. The Bayesian Poisson regression model is known to be challenging to sample from and \cite{dangelo2023canale} propose a novel Metropolis-Hastings proposal built specifically for this model. We employ their algorithm to generate samples from the posterior distribution. The Metropolis-Hastings sampler is available in the \texttt{bpr} \texttt{R}  package \citep{bpr}.

Similar to \cite{dangelo2023canale}, we set $\mathbf{b} = \mathbf{0}$ and $B =  2 I_{23}$.  Over 100 replications, we run the MCMC algorithm for various Monte Carlo lengths and employ different estimators of $\Sigma$ to estimate the posterior mean vector. The true value of both $\bar{g}_n$ and $\Sigma$ is unknown here, so we rely on the estimation of ESS$/n$ and computation time to compare methods. The results are provided in Figure~\ref{fig:poisson}. 

For ESS$/n$, all estimators converge from above, exhibiting critical overestimation and leaving the simulation vulnerable to early termination. This is because the correlation in the Markov chain is persistent enough to not allow for convergence from below. Having said that, the CC-ISE estimator provides the most conservative estimation amongst its competitors. The mISE estimator has significantly larger ESS estimates. The mISE estimator comes with a high computational cost as well, taking as long as 5 minutes to compute the estimator. The CC-MLS estimator is also computationally involved, due to the repeated application of the MLS estimator for all 23 diagonals and does not exhibit better ESS estimation. Comparatively, the CC-ISE estimator is significantly cheaper to compute and yields desirable finite-sample performance.

\begin{figure}

\centering

\includegraphics[width=2.9in]{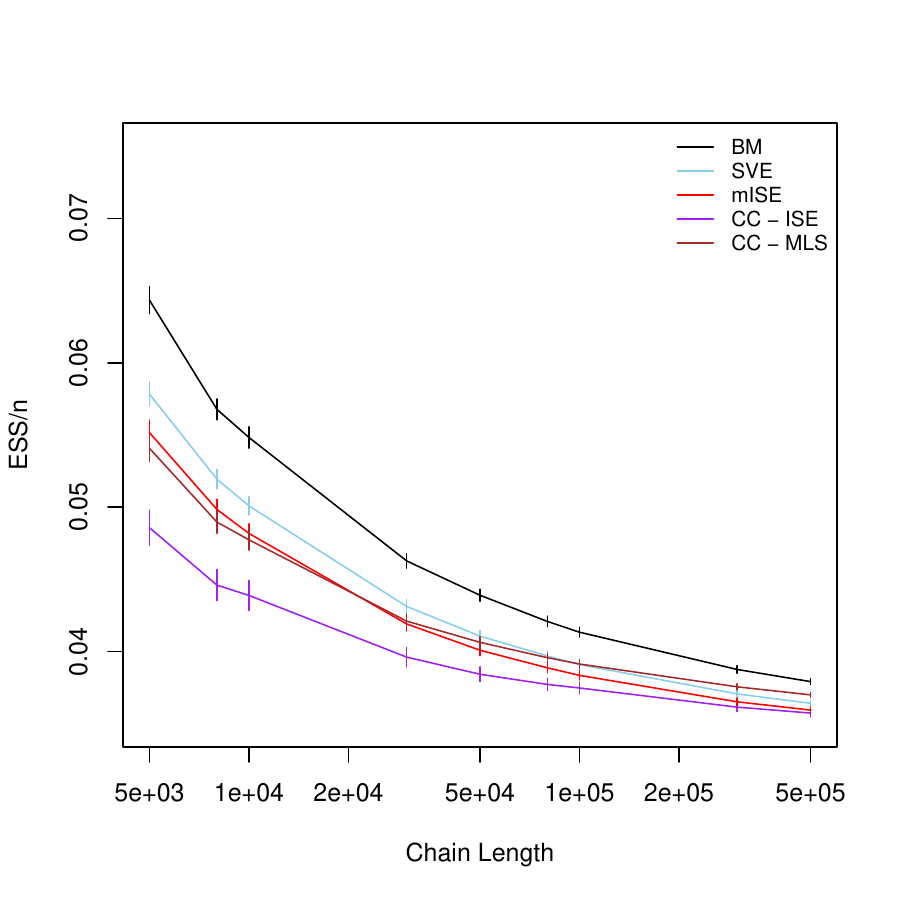}
\includegraphics[width=2.9in]{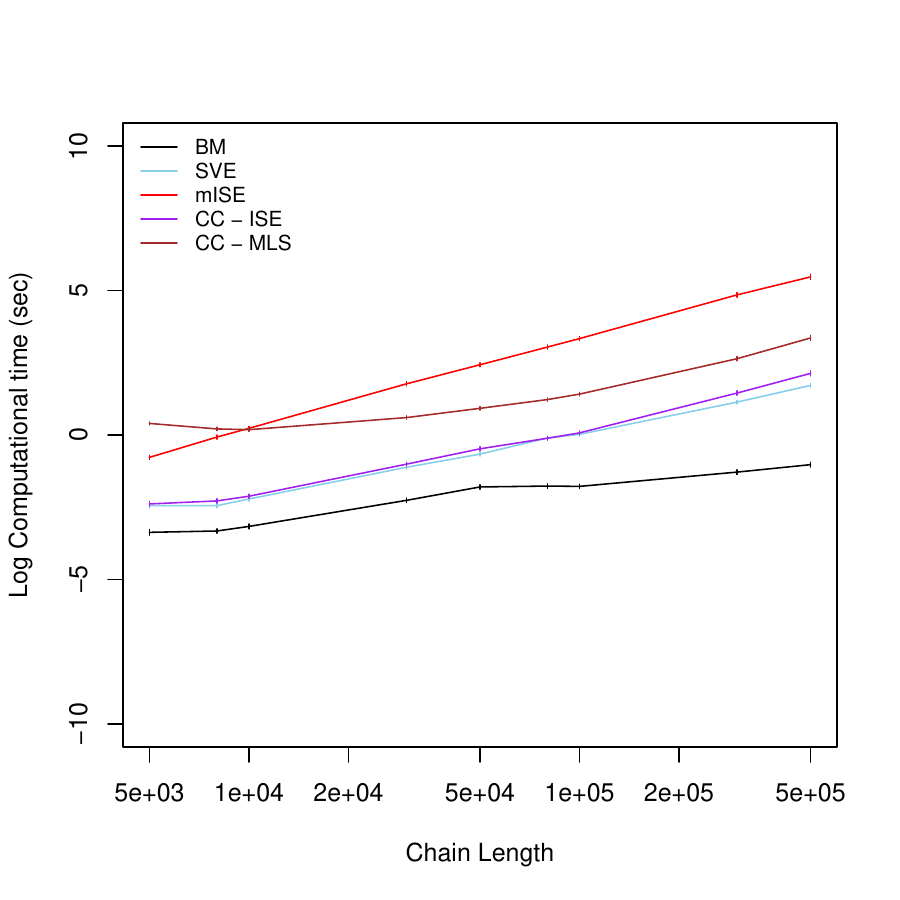}

\caption{(Left) Estimated ESS$/n$ for increasing chain lengths. (Right) Computational time for increasing chain lengths (in seconds).}
\label{fig:poisson}
\end{figure}

\section{Initial sequence estimators for parallel chains}
\label{sec:parallel_ise}

Modern computers have multiple cores that allow for parallel runs of the MCMC to be easily run in the same clock time. Recently, \cite{agarwal2022global,douc2022solving} and \cite{gupta2020estimating} demonstrated the utility of developing specialized methods for estimating $\Sigma$ in such a case. \cite{agarwal2022global} developed globally-centered auto-covariances to estimate $\zeta_{i}$, which we show can be easily adapted to construct a parallel chain version of \possessivecite{geyer1992practical} ISE.

Let $\{X_{t}^{(m)}\}_{1\leq t \leq n}$ be the $m^{th}$ parallel $\pi$-reversible Markov chain for $m = 1, \ldots, M$ generated using the same Markov transition kernel. Consider the case when $d = 1$ and let the Monte Carlo mean for the $m$th chain be $\Bar{g}^{(m)} := n^{-1}\sum_{t=1}^{n} g(X_{t}^{(m)})$ and the global mean be $\Bar{\Bar{g}} := M^{-1} \sum_{m = 1}^{M} \Bar{g}^{(m)}$. \cite{agarwal2022global} define the globally-centered 
 lag-$i$ auto-covariance for the $m^{th}$ parallel chain to be
\begin{equation*}
    \hat{\gamma}_{i;n}^{(m)} = \frac{1}{n} \sum_{t=1}^{n-i} \left(g(X_{t}^{(m)}) - \Bar{\Bar{g}}\right) \left(g(X_{t+i}^{(m)}) - \Bar{\Bar{g}}\right).
\end{equation*}
The final estimator of the auto-covariance is
\begin{equation}
    \label{eq:global_acf}
    \hat{\gamma}_{i;n}^{\text{G}} := \dfrac{1}{M} \sum_{m=1}^{M} \hat{\gamma}_{i;n}^{(m)}\,.
\end{equation}
 \cite{agarwal2022global} demonstrated how a small fix of global-centering can massively improve the estimation of auto-covariance, particularly for slow-mixing Markov chains or multi-modal targets. They specifically demonstrated how globally-centered auto-covariances when employed in spectral-variance estimators, lead to improved estimation of $\sigma^2$ (or $\Sigma$). 

We propose a globally-centered adaptation of the initial sequence estimator of \cite{geyer1992practical}. Let $\hat{\Gamma}_{i;n}^{\text{G}} = (\hat{\gamma}^{\text{G}}_{2i;n} + \hat{\gamma}^{\text{G}}_{2i+1;n})$. We define the globally-centered ISE (G-ISE) as
\begin{equation}
    \hat{\sigma}^{2}_{\text{G-ISE}}  = - \hat{\gamma}^{\text{G}}_{0;n} + 2 \sum_{i=0}^{k_{n}} \hat{\Gamma}^{\text{G}}_{i;n}\,, \label{eq:G-ISE}
\end{equation}
where $k_{n}$ is the largest integer such that $\hat{\Gamma}^{\text{G}}_{i;n} > 0 $ for all $i \in \{1, 2, \ldots, k_{n}\}$. The following theorem establishes the asymptotic conservation result of the GCC-ISE.
\begin{theorem}
\label{thm:gcise_asymp}
    For $\pi$-reversible, Harris ergodic stationary Markov chains such that a univariate CLT holds,
    \begin{equation*}
        \liminf_{n \rightarrow \infty} \hat{\sigma}^{2}_{\text{G-ISE}} \ge \sigma^{2} \qquad \text{with probability $1$}\,.
    \end{equation*}
\end{theorem}
\begin{proof}
    See Appendix~\ref{sec:app_gcise}.
\end{proof}

Having described the univariate globally-centered ISE estimator, we can now present the covariance-correlation globally-centered multivariate estimator. Let $\left(\hat{\sigma}_{\text{G-ISE}}^{(j)} \right)^2$ denote the globally-centered ISE estimator for the $j$th component of $\theta$. Let $\hat{L}_{\text{G-ISE}}$ denote the diagonal matrix of marginal G-ISE standard deviations
    \begin{equation}
    \label{eq:ise_diag_par}
        \hat{L}_{\text{G-ISE}} = \text{diagonal} \left(\hat{\sigma}_{\text{G-ISE}}^{(1)}, \hat{\sigma}_{\text{G-ISE}}^{(2)}, \dots, \hat{\sigma}_{\text{G-ISE}}^{(d)}  \right)\,.
    \end{equation}
\cite{gupta2020estimating} presents a globally-centered version of the batch-means estimator as well, which we denote by  $\hat{\mathcal{R}}_{\text{G-BM}}$. Using $\hat{\mathcal{R}}_{\text{G-BM}}$ and $\hat{L}_{\text{G-ISE}}$, a globally-centered covariance-correlation ISE (GCC-ISE) of $\Sigma$ can be naturally constructed as:
\begin{equation}
    \hat{\Sigma}_{\text{G-cc}} = \hat{L}_{\text{G-ISE}} \hat{\mathcal{R}}_{\text{G-BM}} \hat{L}_{\text{G-ISE}}\,.
\end{equation}
%
%
\begin{corollary}
    \label{cor:global}
    Let $\{X_{t}^{(s)}\}_{t\ge 1}$ be $\pi$-reversible Harris ergodic stationary Markov chains such that a CLT holds for a function $g$. If Assumptions~\ref{assm:1} and \ref{assm:2} hold, we have
    \begin{equation}
        \liminf_{n \rightarrow \infty} \text{det}\left( \hat{\Sigma}_{\text{G-cc}} \right) \ge \text{det}\left( \Sigma \right)   \text{   with probability   } 1.
    \end{equation}
\end{corollary}

\begin{proof}
    See Appendix~\ref{pf:global_multi}
\end{proof}
The popular software STAN \citep{carpenter2017stan} modifies \possessivecite{geyer1992practical} ISE estimator for the parallel chain scenario using an ad-hoc adjustment in the following way.  In the one-dimensional setup, for $1 \leq m \leq M$, denote
\begin{equation*}
    s_{m}^{2} = \frac{1}{(n-1)} \sum_{t=1}^{n} \left(  g(X_{t}^{(m)}) - \Bar{g}^{(m)}  \right)^{2}.
\end{equation*}
Further, denote $\widehat{W} = M^{-1} \sum_{m=1}^{M} s_{m}^{2}$ and 
\begin{equation*}
    \widehat{B} = \frac{n}{(M-1)} \sum_{m=1}^{M} (\Bar{g}^{(m)} - \Bar{\Bar{g}})^{2}.
\end{equation*}
Let $\hat{\gamma}^{(m)}_{i;n}$ denote the lag-$i$ sample auto-covariance for the $m^{th}$ chain. Then the lag-$i$ auto-covariance proposed by STAN is
\begin{equation*}
    \hat{\gamma}_{i;n}^{\text{STAN}} =  \frac{1}{n} (\widehat{B} - \widehat{W}) + \frac{1}{M} \sum_{m=1}^{M} \hat{\gamma}^{(m)}_{i;n}.
\end{equation*}
The initial sequence estimators based on these auto-covariance estimators are denoted as $\hat{\sigma}^{2}_{\text{STAN-ISE}}$, and their theoretical properties have not been studied. Similar to before, we can build the covariance-correlation estimator using STAN auto-covariances for the multivariate settings, i.e., the diagonal matrix with STAN-ISE standard deviations can be defined as 
\begin{equation*}
     \hat{L}_{\text{STAN-ISE}} = \text{diagonal} \left(\hat{\sigma}_{\text{STAN-ISE}}^{(1)}, \hat{\sigma}_{\text{STAN-ISE}}^{(2)}, \dots, \hat{\sigma}_{\text{STAN-ISE}}^{(d)}  \right)\,.
\end{equation*}
The STAN-based multivariate covariance-correlation estimator is then defined as
\begin{equation*}
    \hat{\Sigma}_{\text{STAN-cc}} = \hat{L}_{\text{STAN-ISE}} \hat{\mathcal{R}}_{\text{G-BM}} \hat{L}_{\text{STAN-ISE}}\,.
\end{equation*}
The performance of estimators $\Sigma_{\text{G-cc}}$ and $\Sigma_{\text{STAN-cc}}$ are compared  for a number of different parallel chains $m = 2, 4, 8$ and $16$ for the vector auto-regressive process of Section~\ref{sec:examples}. We compare the relative Frobenius norm using samples of sizes $10^{3}$ and $10^{5}$ with $100$ replications. Results are in Figure~\ref{fig:parallel}. It is clear that although for larger sample sizes $\Sigma_{\text{G-cc}}$ and $\hat{\Sigma}_{\text{STAN-cc}}$ are comparable, for smaller sample sizes, $\hat{\Sigma}_{\text{G-cc}}$ is far more reliable.

%
\begin{figure}
\centering
\includegraphics[width=7.3cm]{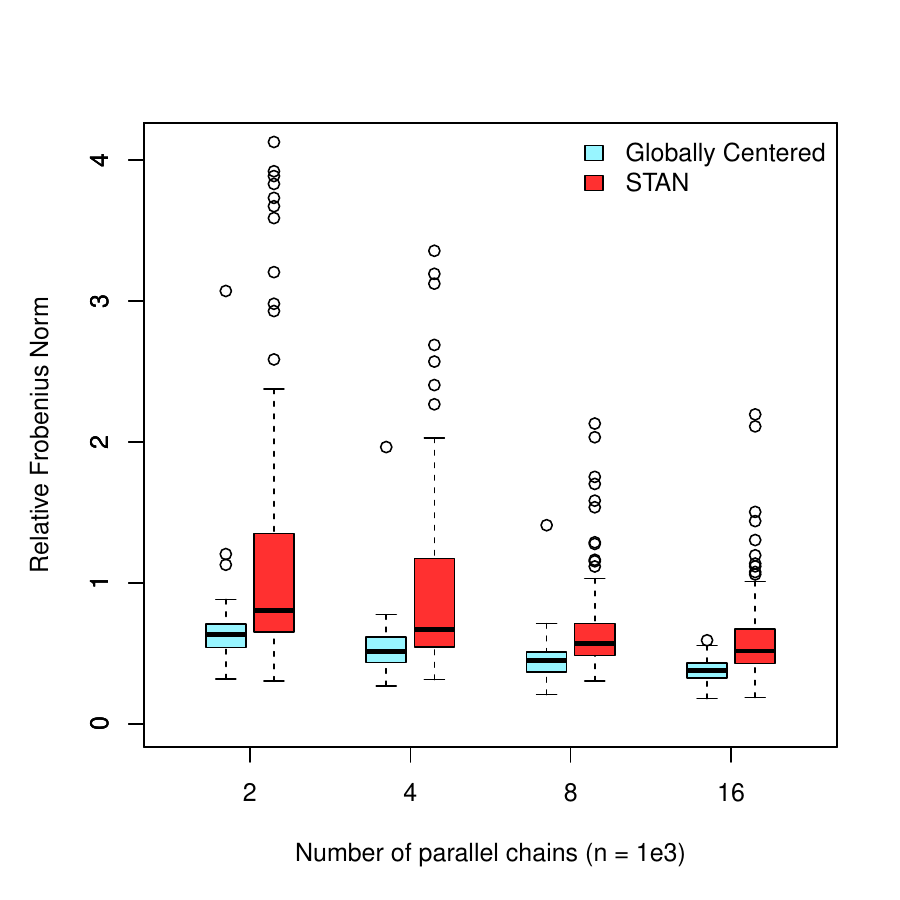}
\includegraphics[width=7.3cm]{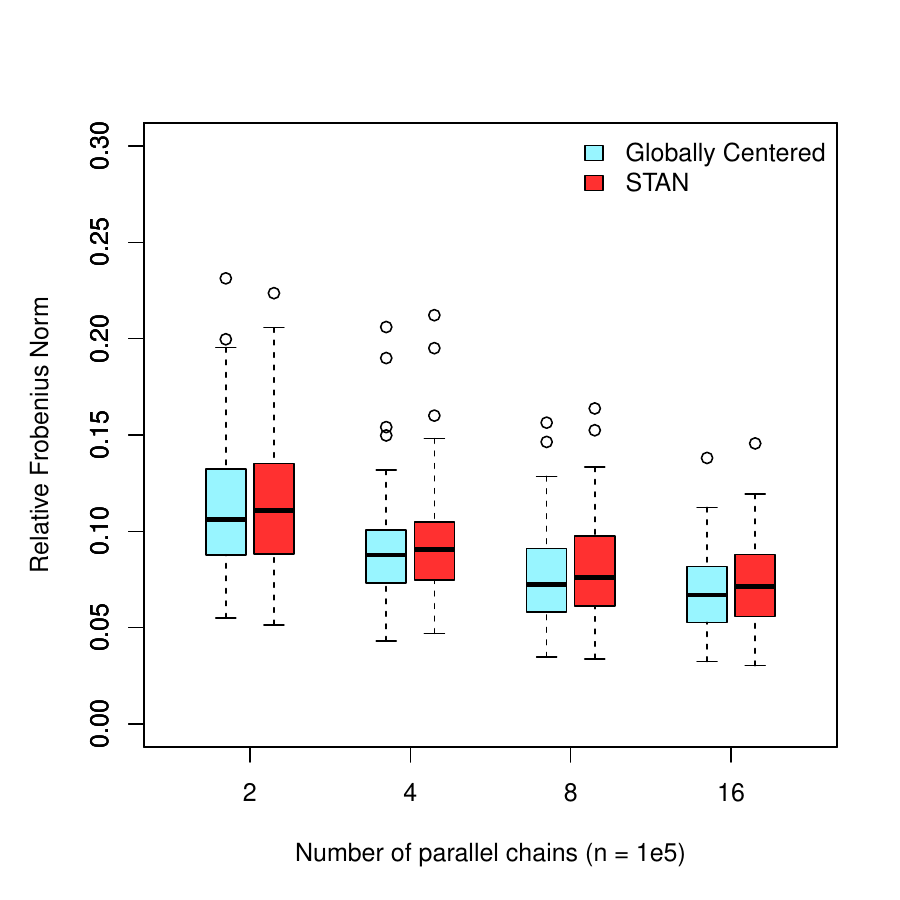}
\caption{Boxplots of estimated relative Frobenius norm for various numbers of parallel chains for (left) $n = 10^3$ and (right) $n = 10^{5}$.}
\label{fig:parallel}
\end{figure}

\section{Discussion}
\label{sec:discussion}

Our central goal was to arrive at a version of a multivariate ISE that is computationally efficient enough to be practically usable and that provides reliable finite-sample results. Since the family of multivariate ISE estimators relies on the estimation of auto-covariance matrices, this has been traditionally challenging. Our solution provides a practical, efficient, and effective alternative to that of \possessivecite{dai2017multivariate} mISE. 

SV estimators also require the estimation of sample lag auto-covariance matrices and, in principle, should be computationally involved. The key difference is that the truncation point is determined \textit{a priori} and not sequentially. In fact, for a fixed truncation point, \cite{heberle2017fast} proposes an algorithm that also relies on discrete Fourier and fast Fourier transforms to yield an efficient implementation of the SV estimator. However, as is evident from our simulations, our proposed alternative yields superior performance at comparative costs. Recently, \cite{song2025multivariate} proposed a multivariate version of their MLS estimator, which requires running their MLS estimator for every element of the $\Sigma$ matrix. Given the computational cost evidenced in Section~\ref{sec:examples} of estimating $d$ diagonals, this estimator is expected to be computationally intensive.

Finally, reversibility of the Markov chain is a critical assumption for employing any initial sequence estimator. It remains an open problem to develop a version of an ISE estimator that is compatible with non-reversible Markov chains.

\section{Acknowledgements}

Dootika Vats is supported by SERB (SPG/2021/001322).

\appendix
\section{Alternative Implementation of ISE}
\label{app:fft}

An alternate way of calculating the ISE estimator of \cite{dai2017multivariate} is to obtain the auto-covariance matrices $\hat{\zeta}_{i;n}$ altogether for all $i$ by element-wise calls to an FFT. In Algorithm~\ref{algo:mise_fft}, we present the details of such an implementation. Compared to Algorithm~\ref{algo:mise} where the complexity was $\mathcal{O}(d^2 n t_n + d^3 t_n)$, Algorithm~\ref{algo:mise_fft} can indeed be faster with a complexity of $\mathcal{O}(d^2 n \log n + d^3 t_n)$. However, in Algorithm~\ref{algo:mise_fft}, the memory burden on the system increases since saving all the auto-covariances in memory costs $\mathcal{O}(nd^{2})$ whereas the original Markov chain requires a memory of only $\mathcal{O}(nd)$. Practically, this can add a significant memory burden on the system. Further, in practice, the $\mathcal{O}(d^3 t_n)$ term remains significant and thus the two implementations are typically not dissimilar. 

\begin{algorithm}[H]
        \caption{mISE (FFT) of \cite{dai2017multivariate}} \label{algo:mise_fft}
        \KwData{$\{g(X_{1}), \ldots, g(X_{n})\}$}
        Calculate $\{\hat{\zeta}_{0;n}, \hat{\zeta}_{1;n}, \ldots, \hat{\zeta}_{n-1;n}\}$; \Comment{Cost: $\mathcal{O}(d^{2}n\log n)$}\\
        $\hat{\Sigma} \gets - \hat{\zeta}_{0;n}$; \\
        \For{$i = 0, 1, \ldots, \lfloor (n-1)/2 \rfloor$}{
            $\hat{\Sigma} = \hat{\Sigma} + 2\hat{Z}_{i;n}$;\\
            \If {\normalfont {$\hat{\Sigma}$ is positive-definite}} {
                $s_{n} = i$;
                \textbf{break};
            } \Comment{Cost: $\mathcal{O}(d^2s_{n})$}
        }
        \For{$i = (s_{n} + 1), \ldots, \lfloor n/2-1 \rfloor$}{
            $S = \hat{\Sigma}$;\\
            $\hat{\Sigma} = \hat{\Sigma} + 2\hat{Z}_{i;n}$;\\
            Calculate \text{det}($\hat{\Sigma}$) and \text{det}($\hat{S}$)  \Comment{Cost: $\mathcal{O}(d^3)$}\\
            \If {\normalfont {\text{det}($\hat{\Sigma}$) $\leq$ \text{det}($S$)}}{
                $t_{n} = (i-1)$;   \Comment{The truncation time} \\ 
                $\hat{\Sigma}_{\text{ISE}} \gets S$;
                \textbf{break};
            } \Comment{Cost: $\mathcal{O}(d^{3}t_{n})$}
        }
        \KwRet {\normalfont {$\hat{\Sigma}_{\text{ISE}}$}} \Comment{Total Cost: $\mathcal{O}(d^{2} n \log n + d^3t_n)$}\\
    \end{algorithm}

\section{Estimation Quality of Batch-Means Correlations}
\label{sec:underestimation}
For $\Sigma$ being the true asymptotic covariance matrix and $\sigma_{ij}$s to be the elements of $\Sigma$, the elements of the correlation matrix, $\mathcal{R}$, are given by
\begin{align*}
    r_{ij} = f(\sigma_{ij}, \sigma_{ii}, \sigma_{jj}) = \frac{\sigma_{ij}}{\sqrt{\sigma_{ii} \ \sigma_{jj}}}\, .
\end{align*}
For the batch-means estimator $\widehat{\Sigma}_{\rm BM} = \{\hat{\sigma}_{ij}\}$, the components of the estimated correlation matrix are
\begin{align*}
    \hat{r}_{ij} = f(\hat{\sigma}_{ij}, \hat{\sigma}_{ii}, \hat{\sigma}_{jj}) = \frac{\hat{\sigma}_{ij}}{\sqrt{\hat{\sigma}_{ii} \ \hat{\sigma}_{jj}}}\,.
\end{align*}
By a Taylor's series expansion of $\hat{r}_{ij}$ around $r_{ij}$,
\begin{align*}
    \hat{r}_{ij} & = r_{ij} + (\hat{\sigma}_{ij} - \sigma_{ij}, \hat{\sigma}_{ii} - \sigma_{ii}, \hat{\sigma}_{jj} - \sigma_{jj}) \nabla f(\sigma_{ij}, \sigma_{ii}, \sigma_{jj}) + \ldots \\
    & \approx r_{ij} + (\hat{\sigma}_{ij} - \sigma_{ij}, \hat{\sigma}_{ii} - \sigma_{ii}, \hat{\sigma}_{jj} - \sigma_{jj}) \nabla f(\sigma_{ij}, \sigma_{ii}, \sigma_{jj}) \\
    \Rightarrow \text{Bias}(\hat{r}_{ij}) & \approx \left(\text{Bias}(\hat{\sigma}_{ij}), \text{Bias}(\hat{\sigma}_{ii}), \text{Bias}(\hat{\sigma}_{jj})\right) \nabla f(\sigma_{ij}, \sigma_{ii}, \sigma_{jj})\\
    & = \frac{\text{Bias}(\hat{\sigma}_{ij})}{\sqrt{\sigma_{ii} \ \sigma_{jj}}} - \frac{r_{ij}}{2} \frac{\text{Bias}(\hat{\sigma}_{ii})}{\sigma_{ii}} - \frac{r_{ij}}{2} \frac{\text{Bias}(\hat{\sigma}_{jj})}{\sigma_{jj}}\\
    \Rightarrow \frac{\text{Bias}(\hat{r}_{ij})}{r_{ij}} & \approx \frac{\text{Bias}(\hat{\sigma}_{ij})}{\sigma_{ij}} - \frac{1}{2}\left( \frac{\text{Bias}(\hat{\sigma}_{ii})}{\sigma_{ii}} +  \frac{\text{Bias}(\hat{\sigma}_{jj})}{\sigma_{jj}} \right)\,.
\end{align*}
The above provides a relationship between the relative bias of the batch-means correlation estimator and the relative bias of the batch-means covariance matrix estimator. Let (assuming $\pi$-reversibility),
\[
\Xi =  2 \sum_{k=1}^{\infty} k \text{Cov}(X_1, X_{1+k}) \,.
\]
Then under regularity conditions, with $b_n$ being the batch-size of the batch-means estimator, \cite{vats2022lugsail} show that
\[
\text{Bias} \left( \hat{\Sigma}_{\rm BM} \right) \approx -\dfrac{\Xi}{b_n}\,.
\]
Consider matrix $B = -\Xi/\Sigma$ (where the division is element-wise) and let $B_{ij}$ denote the $(i,j)^\text{th}$ component of $B$. For positive Markov chains, the diagonals of $B$ are negative \citep{vats2022lugsail}. It is then fair to say that the relative bias of the batch-means correlation estimator will be better when
\[
|B_{ij}| > \dfrac{B_{ii} + B_{ij}}{2}\,.
\]
A theoretical justification explaining the conditions under which the above relationship holds is challenging to show, and we leave this as an open problem. However, for the VAR example, we note that the true $\Sigma$ is known. This allows us to study the behavior of the batch-means correlation and covariance matrix estimators. To compare the quality of estimation of $r_{ij}$ versus $\sigma_{ij}$, we run a two-dimensional VAR process $\{X_t\}$ with $\Phi$ being a symmetric matrix and $\rho > 0 $ such that 
\begin{equation*}
    X_t = \rho \Phi X_{t-1} + \epsilon_t; \ \epsilon_t \overset{\rm iid}{\sim} {\rm Normal}(0, I_2).
\end{equation*}
For the process to exhibit a stationary and limiting distribution, $\Phi$ must have a spectral radius less than 1, and a higher spectral radius corresponds to slower mixing in the VAR process \citep{dai2017multivariate}. In order to obtain such a $\Phi$, we generate a $2 \times 2$ matrix $A$ using iid draws from $N(0,1)$ and set $B = A A^{\top}$. Let $\eta$ denote the largest eigenvalue of $B$ and set $\Phi  = B/(\eta + 0.001)$. This ensures that the spectral radius of $\rho\Phi \approx \rho$. We generate two such $\Phi$ for two different seeds.

For varying $\rho$ and over $1000$ replications, we plot the estimates of the absolute relative deviation of the $(1, 2)^{{\text{th}}}$ element of the batch-means correlation matrix ($R$) and the covariance matrix ($\Sigma$) in Figures~\ref{fig:COV_CORR_ABS_BIAS_1} and \ref{fig:COV_CORR_ABS_BIAS_2}. The absolute relative deviation is defined as 
\begin{align*}
    \text{Covariance Absolute Relative Deviation} & = \frac{\mathbb{E}\left\vert\widehat{\sigma}_{12} - \sigma_{12}\right\vert}{ \sigma_{12} }  \\
    \text{Correlation Absolute Relative Deviation} & = \frac{\mathbb{E}\left\vert\widehat{r}_{12} - r_{12}\right\vert}{ r_{12} } \,.
\end{align*}
Similarly, the relative variances at the $(1, 2)^{\text{th}}$ position are
\begin{align*}
        \text{Covariance Relative Variance} & = \Var\left( \frac{\hat{\sigma}_{12}}{\sigma_{12}} \right)\\
    \text{Correlation Relative Variance} & = \Var\left( \frac{\hat{r}_{12}}{r_{12}} \right)\,.
\end{align*}
The plots show that for all values of $r_{12}$, the correlation estimator exhibits a smaller absolute relative deviation and also a smaller relative variance. Further, we see an interesting trend that for larger magnitudes in $r_{12}$ (connected also to higher $\rho$ values), the absolute relative deviation and the relative variance for the correlation estimator are decreasing, while the estimation quality for the covariance batch-means estimator gets worse. A theoretical exploration of the properties of the correlation batch-means estimator will make for interesting future work. 
\begin{figure}[ht]
    \centering
    \begin{minipage}{0.49\linewidth}
        \centering
        \includegraphics[width=\linewidth]{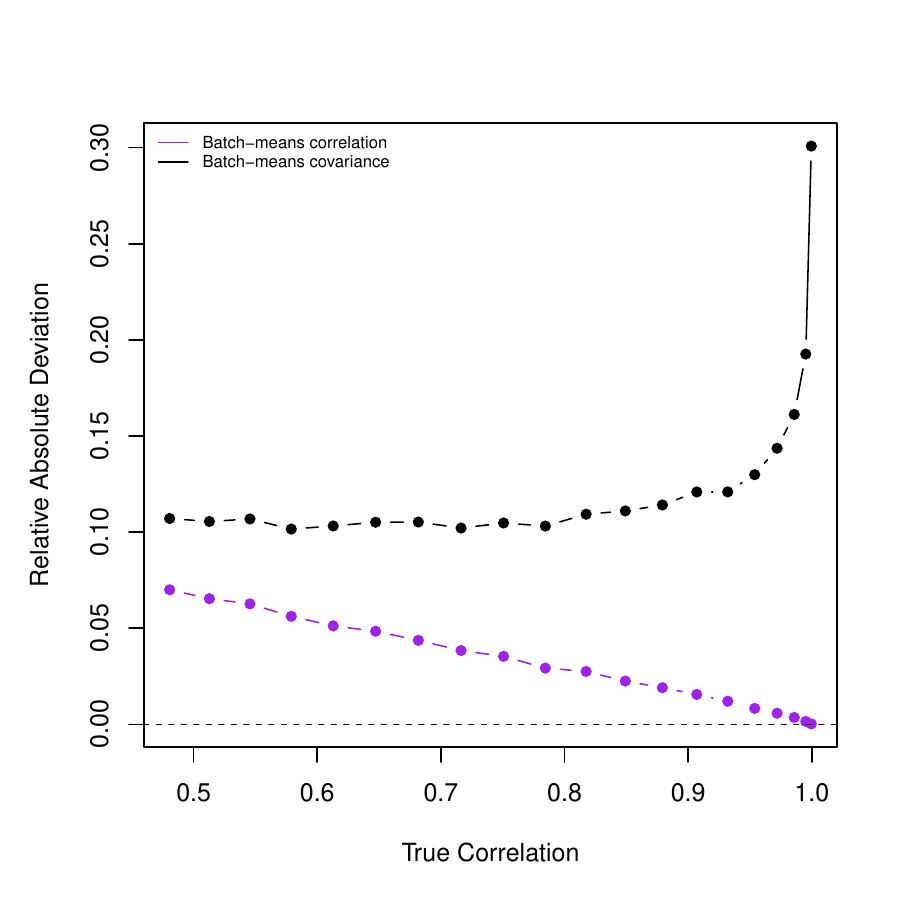}
    \end{minipage}
    \hfill
    \begin{minipage}{0.49\linewidth}
        \centering
        \includegraphics[width=\linewidth]{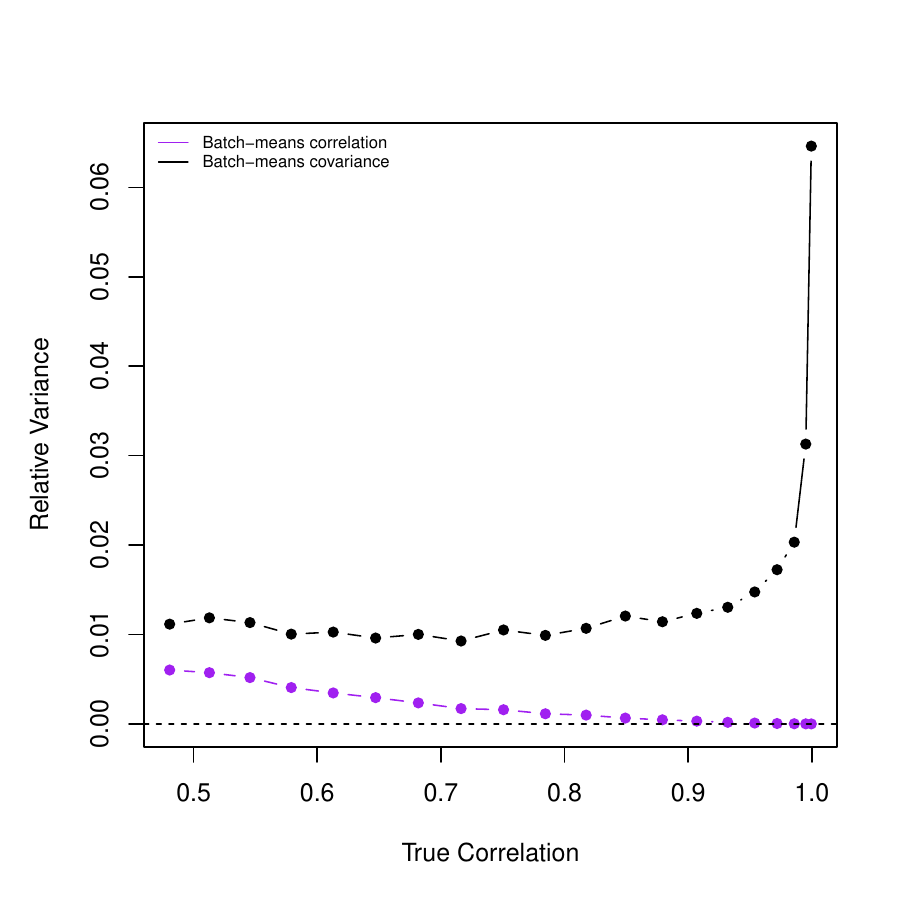}
    \end{minipage}
    \caption{Absolute relative deviation for batch-means correlation and batch-mean covariance for the first seed.}
    \label{fig:COV_CORR_ABS_BIAS_1}
\end{figure}
\begin{figure}[ht]
    \centering
    \begin{minipage}{0.49\linewidth}
        \centering
        \includegraphics[width=\linewidth]{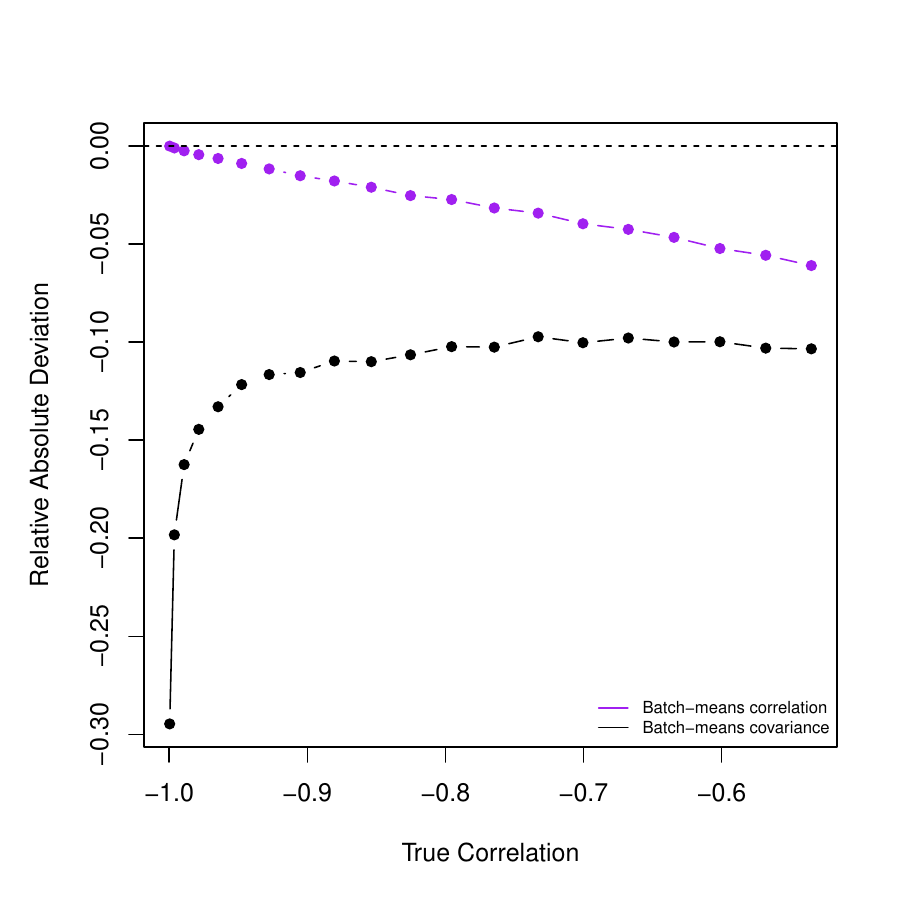}
    \end{minipage}
    \hfill
    \begin{minipage}{0.49\linewidth}
        \centering
        \includegraphics[width=\linewidth]{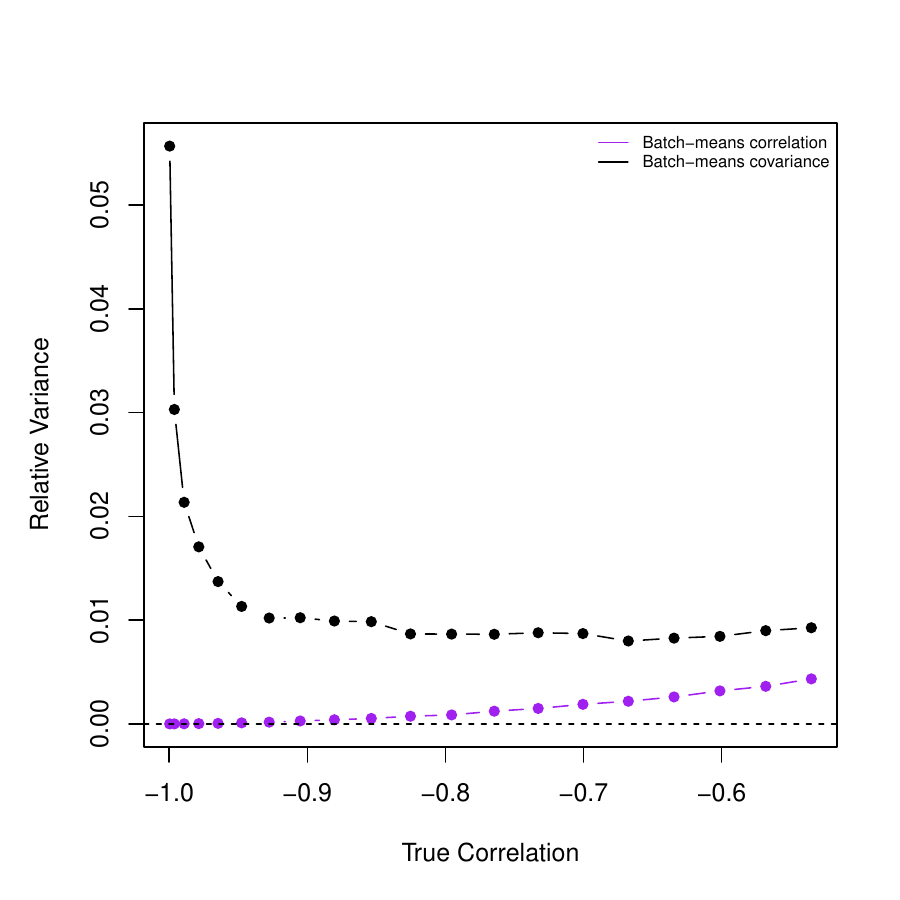}
    \end{minipage}
    \caption{Absolute relative deviation for batch-means correlation and batch-mean covariance for the second seed.}
    \label{fig:COV_CORR_ABS_BIAS_2}
\end{figure}

\section{Theoretical results}
\subsection{Proof of Proposition~\ref{thm:mise_tn}}
\label{pf:mise_tn}

\begin{theorem}{\cite[][Proposition~$2$]{dai2017multivariate}}
    Let $\Sigma_{m} = -\zeta_{0} + 2 \sum_{i=0}^{m} Z_{i}$ where $Z_{m} = \zeta_{2m} + \zeta_{2m + 1}$ for all $m \ge 0$. Then there exists an integer $m_{0}$ such that
    \begin{itemize}
        \item[1.] $\Sigma_{m}$ is positive definite for all $m \ge m_{0}$ and not positive definite for all $m < m_{0}$; and
        \item[2.] the sequence $\{\vert \Sigma_{m} \vert : m \ge m_{0}\}$ is positive, increasing, and converges to $\vert \Sigma \vert$.
    \end{itemize}
\end{theorem}

\begin{proof}[Proof of Proposition~\ref{thm:mise_tn}]
    Define $\hat{\Delta}_{i;n} = \vert \hat{\Sigma}_{i;n} \vert - \vert \hat{\Sigma}_{i-1;n} \vert$ and $\Delta_{i} = \vert \Sigma_{i} \vert - \vert \Sigma_{i-1} \vert$ for all $i \ge 1$. From ) in \citet[][Equation (B5)]{dai2017multivariate}, for all $i$, $\hat{\Delta}_{i;n} \rightarrow \Delta_{i}$ as $n \rightarrow \infty$ with probability $1$. Define the set of events $A_{t}$ for all $t \ge 0$ such that 
    \begin{equation*}
        A_{t} := \left\{\lim_{n \rightarrow \infty} \hat{\Delta}_{t;n} = \Delta_{t} \right\}.
    \end{equation*}
    From \cite{dai2017multivariate}, 
    \begin{equation*}
        \Pr(A_{t}^{c}) = 0 \numberthis \label{eq: A_t}\,,
    \end{equation*}
    for all $t \ge 0$. Define event $B = \{s_{n} \leq M \text{ for some } M \in \mathbb{N}\}$. Then by \citet[][Theorem~$2$]{dai2017multivariate},
    \begin{equation}
    \label{eq: B}
        \Pr(B) = 1 \Rightarrow \Pr(B^{c}) = 0
    \end{equation}
    By the Boole's inequality 
    \begin{align*}
        &\Pr\left((\cup_{t = s_{n} + 1}^{\infty}A_{t}^{c}) \cup B^{c}\right) \leq  \sum_{t=s_{n} + 1}^{\infty} \Pr(A_{t}^{c}) + \Pr(B^{c})\\
        \Rightarrow &\Pr\left((\cap_{t = s_{n} + 1}^{\infty}A_{t} \cap B)^{c}\right) \leq  \sum_{t = s_{n} + 1}^{\infty} \Pr(A_{t}^{c}) + \Pr(B^{c})\\
        \Rightarrow & 1 - \Pr(\cap_{t = s_{n} + 1}^{\infty}A_{t} \cap B) \leq 0 \ (\text{by } \eqref{eq: A_t} \text{ and } \eqref{eq: B})\\
        \Rightarrow & \Pr(\cap_{t = s_{n} + 1}^{\infty}A_{t} \cap B) = 1\\
        \Rightarrow & \Pr\left(\lim _{n \rightarrow \infty} \hat{\Delta}_{i;n} = \Delta_{i} \quad \text{ for all } i \ge s_{n} + 1\right) = 1. \numberthis \label{eq:}
    \end{align*}
    Now
    \begin{align*}
        & \Pr(t_{n} \rightarrow \infty \text{ as }  n \rightarrow \infty) \\
        = & \Pr\left(\vert \hat{\Sigma}_{i;n} \vert > \vert \hat{\Sigma}_{i-1;n} \vert \text{ as } n \rightarrow \infty \text{ for all } i \ge s_{n} + 1 \right) \\
        = & \Pr\left(\vert \hat{\Sigma}_{i;n} \vert - \vert \hat{\Sigma}_{i-1;n} \vert > 0 \text{ as } n \rightarrow \infty \text{ for all } i \ge s_{n} + 1 \right) \\
        \ge & \Pr\left(\hat{\Delta}_{i;n} \rightarrow \Delta_{i} \text{ as } n \rightarrow \infty \text{ for all } i \ge s_{n} + 1\right) \\
        = & 1.
    \end{align*}
\end{proof}

\subsection{Proof of Theorem~\ref{thm:Gen_var}}
\begin{proof}
\label{pf:Gen_var}
We show each part separately.
\begin{itemize}
\item[1.] Note that
    \begin{align*}
        \vert \hat{\Sigma}_{\text{cc}} \vert & = \vert \hat{L}_{\text{ISE}} \vert \vert \hat{L}^{\top}_{\text{ISE}} \vert \vert \hat{\mathcal{R}}_{\text{BM}} \vert \\
        & = \left(\prod_{i=1}^{d} \left(\hat{\sigma}_{\text{ISE}}^{(i)}\right)^{2}\right) \vert \hat{\mathcal{R}}_{\text{BM}} \vert.
    \end{align*}
     Under the Assumption~\ref{assm:1} and \ref{assm:2} and from \cite{vats2019multivariate} Theorem~$2$, $\hat{\Sigma}_{\text{BM}} \rightarrow \Sigma$ as  $n \rightarrow \infty$ with probability $1$. Also define $\hat{a}_{ij; \text{BM}}$ to be the $(i, j)^{th}$ elements of $\hat{\Sigma}_{\text{BM}}$. Consequently, $\hat{a}_{ij; \text{BM}} \rightarrow a_{ij}$ as $n \rightarrow \infty$ with probability $1$. For $\hat{a}_{ii; \text{BM}} > 0$ for all $i$ and by continuous mapping theorem, with probability 1
     \begin{align*}
         \frac{\hat{a}_{ij; \text{BM}}}{\sqrt{\hat{a}_{ii; \text{BM}} \ \hat{a}_{jj; \text{BM}}}} & \rightarrow \frac{a_{ij}}{\sqrt{a_{ii} \  a_{jj}}} \ \text{as } n \rightarrow \infty \\
         \Rightarrow \hat{r}_{ij} & \rightarrow r_{ij} \ \text{as } n \rightarrow \infty
     \end{align*}
     where $\hat{r}_{ij}$ and $r_{ij}$ are $(i, j)^{th}$ elements of $\hat{\mathcal{R}}_{\text{BM}}$ and $\mathcal{R}$ respectively. Hence as $n \rightarrow \infty$
     \begin{equation}
         \hat{\mathcal{R}}_{\text{BM}} \rightarrow \mathcal{R} \text{ with probability } 1. \label{eq: corr_cov}
     \end{equation}
     By \citet[][Theorem~$3$]{vats2018strong}, the eigenvalues of $\hat{\mathcal{R}}_{\text{BM}}$ are strongly consistent to eigenvalues of $\mathcal{R}$. This implies as  $n \rightarrow \infty$ with probability $1$
     \begin{equation}
         \vert \hat{\mathcal{R}}_{\text{BM}} \vert \rightarrow \vert \mathcal{R} \vert. \label{eq:det_corr}
     \end{equation}
     Additionally, from \citet[][Theorem~$3$]{geyer1992practical},  as  $n \rightarrow \infty$ with probability $1$ for $i = 1, 2, \cdots, d$
     \begin{equation}
         \liminf_{n \rightarrow \infty} \left(\hat{\sigma}_{\text{ISE}}^{(i)}\right)^{2} \ge \left(\sigma^{(i)}\right)^{2}.\label{eq:geyer_liminf}
     \end{equation}
     Hence, from \eqref{eq:det_corr} and \eqref{eq:geyer_liminf} with probability $1$
     \begin{align*}
        \liminf_{n \rightarrow \infty} \vert \hat{\Sigma}_{\text{cc}} \vert & = \left( \prod_{i=1}^{d} \liminf_{n \rightarrow \infty} \left(\hat{\sigma}_{\text{ISE}}^{(i)}\right)^{2}\right) \lim_{n \rightarrow \infty} \vert \hat{\mathcal{R}}_{\text{BM}}\vert \\
        & \ge \left( \prod_{i=1}^{d} \left(\sigma^{(i)}\right)^{2} \right) \vert \mathcal{R}\vert \\
        & = \vert L \mathcal{R} L \vert \\
        & = \vert \Sigma \vert.
     \end{align*}

\item[2.] From \eqref{eq: corr_cov}, as $n \rightarrow \infty$
\begin{align*}
    \hat{\mathcal{R}}_{\text{BM}} & \overset{\text{a.s.}}{\rightarrow} \mathcal{R}  \Rightarrow \hat{r}_{ij} \rightarrow r_{ij} \text{ for all } (i,j) \text{ with probability } 1.
\end{align*}
Hence with probability $1$
\begin{align*}
    \liminf_{n \rightarrow \infty} \vert \hat{a}_{ij} \vert & = \liminf_{n \rightarrow \infty} \left( \hat{\sigma}_{\text{ISE}}^{(i)} \vert \hat{r}_{ij} \vert \hat{\sigma}_{\text{ISE}}^{(j)}\right) \\
    & \ge \sigma^{(i)} \vert r_{ij} \vert \sigma^{(j)}\\
    &  = \vert a_{ij} \vert.
\end{align*}


\end{itemize}
\end{proof}

\subsection{Proof of Theorem~\ref{thm:gcise_asymp}}
\label{sec:app_gcise}
\begin{proof}
 \label{pf:global}
 First, we show that globally centered auto-covariance matrices, averaged over all the parallel chains are strongly consistent for the true auto-covariance matrix. Hence for all $t \ge 0$ and for any initial distribution, as $n \rightarrow \infty$
    \begin{align*}
         \hat{\gamma}_{t;n}^{\text{G}} & = \frac{1}{M} \sum_{s = 1}^{M} \hat{\gamma}^{(s)}_{t;n} \\
         & = \frac{1}{M} \sum_{s = 1}^{M} \frac{1}{n-t} \sum_{i = 1}^{n-t} \left(g(X_{t}^{(s)}) - \Bar{\Bar{g}}\right) \left(g(X_{t+i}^{(s)}) - \Bar{\Bar{g}}\right) \\
         & = \frac{1}{M} \sum_{s = 1}^{M} \frac{1}{n-t} \sum_{i = 1}^{n-t} \left[g(X_{t}^{(s)}) g(X_{t+i}^{(s)}) - g(X_{t}^{(s)}) \Bar{\Bar{g}} - \Bar{\Bar{g}} g(X_{t+i}^{(s)}) + \Bar{\Bar{g}}^{2} \right] \\
          & \overset{\text{a.s.}}{\rightarrow} \frac{1}{M} \sum_{s = 1}^{M} \left[\E_{\pi}(g(X_{1}) g(X_{1+t})) - \left(\E_{\pi}(g(X_{1})\right)^{2}\right] \\
         & \overset{\text{a.s.}}{\rightarrow} \E_{\pi}(g(X_{1}) g(X_{1+t})) - \left(\E_{\pi}(g(X_{1})\right)^{2} \\
         & \overset{\text{a.s.}}{\rightarrow} \gamma_{t}.
    \end{align*}
    Next, recall that the assumption of the existence of a univariate central limit theorem guarantees that there exists $\sigma^2 <\infty$ such that $\sigma^2 = -\gamma(0) + 2\sum_{t=0}^{\infty} \gamma_{t}$ \citep[see][]{jones2004markov}. Thus,  the series $\sum_{t=0}^{\infty} \gamma_{t}$ converges. Now the globally-centered univariate ISE is defined as
    \begin{equation*}
        \hat{\sigma}_{\text{G-ISE}}^{2} = - \hat{\gamma}_{0;n}^{\text{G}} + 2 \sum_{i = 0}^{k_{n}} (\hat{\gamma}_{2t;n}^{\text{G}} + \hat{\gamma}_{2t + 1;n}^{\text{G}})
    \end{equation*}
    where $k_{n}$ is the stopping time as defined in \eqref{eq:G-ISE}. Note that for every $\epsilon > 0$, there exists a stopping time $m_{\epsilon} > 0$ such that
    \begin{align*}
        2 \sum_{i = m_{\epsilon}+1}^{\infty} (\gamma_{2t} + \gamma_{2t + 1}) & < \epsilon/2 \\
        \Rightarrow - \gamma_{0} + 2 \sum_{i = 0}^{m_{\epsilon}} (\gamma_{2t} + \gamma_{2t + 1}) & > \sigma^{2} - \epsilon/2.\numberthis \label{eq:par_ineq}
    \end{align*}
    Again, for every $\epsilon > 0$, there exists a sample size $n_{\epsilon} > 0$ such that for all $n > n_{\epsilon}$, and $t \in \{0, 1, \ldots, m_{\epsilon}\}$, $(\hat{\gamma}_{2t;n}^{\text{G}} + \hat{\gamma}_{2t + 1;n}^{\text{G}}) > 0$ and
    \begin{align*}
        & - \gamma_{0} + 2 \sum_{i = 0}^{m_{\epsilon}} (\gamma_{2t} + \gamma_{2t + 1}) - \left( - \hat{\gamma}_{0;n}^{\text{G}} + 2 \sum_{i = 0}^{m_{\epsilon}} (\hat{\gamma}_{2t;n}^{\text{G}} + \hat{\gamma}_{2t + 1;n}^{\text{G}})\right) < \epsilon/2 \\
        \Rightarrow & - \hat{\gamma}_{0;n}^{\text{G}} + 2 \sum_{i = 0}^{m_{\epsilon}} (\hat{\gamma}_{2t;n}^{\text{G}} + \hat{\gamma}_{2t + 1;n}^{\text{G}}) > - \gamma_{0} + 2 \sum_{i = 0}^{m_{\epsilon}} (\gamma_{2t} + \gamma_{2t + 1}) - \epsilon/2\\
        \Rightarrow & - \hat{\gamma}_{0;n}^{\text{G}} + 2 \sum_{i = 0}^{m_{\epsilon}} (\hat{\gamma}_{2t;n}^{\text{G}} + \hat{\gamma}_{2t + 1;n}^{\text{G}})  > \sigma^{2} - \epsilon \text{  (by \eqref{eq:par_ineq})}\\
        \Rightarrow & \inf_{n > n_{\epsilon}} \left( - \hat{\gamma}_{0;n}^{\text{G}} + 2 \sum_{i = 0}^{m_{\epsilon}} (\hat{\gamma}_{2t;n}^{\text{G}} + \hat{\gamma}_{2t + 1;n}^{\text{G}})\right) > \sigma^{2} - \epsilon.
    \end{align*}
    Note that as $\epsilon \rightarrow 0$, $n_{\epsilon} \rightarrow \infty$ and the stopping time $m_{\epsilon} \rightarrow \infty$. Also by definition, for all $n > n_{\epsilon}$, and $t \in \{0, 1, \ldots, m_{\epsilon}\}$, $(\hat{\gamma}_{2t;n}^{\text{G}} + \hat{\gamma}_{2t + 1;n}^{\text{G}}) > 0$ and it implies that $k_{n} \ge m_{\epsilon}$. Hence
    \begin{align*}
        & \liminf_{n \rightarrow \infty} \left( - \hat{\gamma}_{0;n}^{\text{G}} + 2 \sum_{i = 0}^{m_{\epsilon}} (\hat{\gamma}_{2t;n}^{\text{G}} + \hat{\gamma}_{2t + 1;n}^{\text{G}})\right) \ge \sigma^{2}\\
        & \liminf_{n \rightarrow \infty} \left( - \hat{\gamma}_{0;n}^{\text{G}} + 2 \sum_{i = 0}^{k_{n}} (\hat{\gamma}_{2t;n}^{\text{G}} + \hat{\gamma}_{2t + 1;n}^{\text{G}})\right) \ge \sigma^{2} \text{   (since } k_{n} \ge m_{\epsilon})\\
        \Rightarrow & \liminf_{n \rightarrow \infty} \hat{\sigma}^{2}_{\text{G-ISE}} \ge \sigma^{2} \numberthis \label{eq:uni_par_ise}.
    \end{align*}
\end{proof}
\subsection{Proof of Corollary~\ref{cor:global}}
\begin{proof}
\label{pf:global_multi}
    Let $\{X_{t}^{(s)}\}_{t\ge 1}$ be a set of parallel Markov chains for $s = 1, 2, \ldots, M$. By Theorem~\ref{thm:gcise_asymp} as $n \rightarrow \infty$ and for all $j = 1, 2, \ldots, M$
    \begin{equation*}
        \liminf_{n \rightarrow \infty} \left(\hat{\sigma}_{\text{G-ISE}}^{(j)}\right)^{2} \ge \left(\sigma^{(j)}\right)^{2}.
    \end{equation*}
Hence
\begin{equation*}
    \liminf_{n \rightarrow \infty} \text{det}\left( \hat{L}_{\text{G-ISE}} \right) = \prod_{j=1}^{d} \hat{\sigma}_{\text{G-ISE}}^{(j)} \ge \text{det}\left( L \right). \numberthis \label{eq:det_L_par}
\end{equation*}
  Also, under Assumptions~\ref{assm:1} and \ref{assm:2}, by \citet[][Theorem~$2$]{gupta2020estimating}, 
    \begin{align*}
        \lim_{n \rightarrow \infty} \hat{\Sigma}_{\text{G-BM}} & = \Sigma \text{   with probability   } 1  \\
        \Rightarrow \lim_{n \rightarrow \infty} \hat{\mathcal{R}}_{\text{G-BM}} & = \mathcal{R} \text{   with probability   } 1 .
    \end{align*}
Hence, by \citet[][Theorem~$3$]{vats2018strong}, the eigenvalues of $\hat{\mathcal{R}}_{\text{G-BM}}$ are strongly consistent to the eigenvalues of $\mathcal{R}$. This implies as  $n \rightarrow \infty$
    \begin{equation*}
        \lim_{n \rightarrow \infty} \text{det}\left( \hat{\mathcal{R}}_{\text{G-BM}} \right) = \text{det}\left( \mathcal{R} \right) \text{   with probability   } 1 \numberthis \label{eq:R_bm_par}.
    \end{equation*}
By \eqref{eq:det_L_par} and \eqref{eq:R_bm_par}
\begin{equation*}
    \liminf_{n \rightarrow \infty} \text{det}\left( \hat{\Sigma}_{\text{G-cc}} \right) = \liminf_{n \rightarrow \infty} \text{det}\left( \hat{L}_{\text{G-ISE}} \hat{\mathcal{R}}_{\text{G-BM}} \hat{L}_{\text{G-ISE}} \right) \ge \text{det}\left( L \mathcal{R} L \right) = \text{det}\left( \Sigma \right).
\end{equation*}

\end{proof}

\begin{center}
{\large\bf SUPPLEMENTARY MATERIAL}
\end{center}

\begin{description}

\item[Title:] R-codes (zip file)

\item[R-codes for numerical implementations:] This zip file contains \texttt{R}-codes to implement the Covariance correlation estimator as shown in Section~\ref{sec:examples} and \ref{sec:parallel_ise} in the article. The codes are alternatively available in a public access GitHub repository at\\ \href{https://github.com/Arkagit/Efficient-Initial-Sequence-estimator}{https://github.com/Arkagit/Efficient-Initial-Sequence-estimator}.
\end{description}

\bibliographystyle{apalike}

\bibliography{efficientISE_arxiv}
\end{document}